\documentclass[11pt]{article}
\usepackage[margin=.9in]{geometry}
\usepackage{graphicx}
\usepackage{textcomp}
\usepackage{xcolor}
\usepackage{makecell}

\usepackage{fancyvrb}
\usepackage{algorithm}
\usepackage[noend]{algpseudocode}
\usepackage{hyperref}
\usepackage{booktabs}
\usepackage{subcaption}
\usepackage{cite}
\usepackage{authblk}
\usepackage{amsmath,amsthm,amssymb,amsfonts}
\usepackage{listings}
\usepackage[normalem]{ulem}
\usepackage{appendix}
\usepackage{hyperref}
\usepackage{pdfpages}
\usepackage{booktabs}
\usepackage{caption}
\usepackage{subcaption}
\usepackage{multirow}
\usepackage{mathtools} 
\def\BibTeX{{\rm B\kern-.05em{\sc i\kern-.025em b}\kern-.08em
    T\kern-.1667em\lower.7ex\hbox{E}\kern-.125emX}}

\newtheorem{lemma}{Lemma}
\newtheorem{theorem}{Theorem}

\newcommand{\etal}{\emph{et. al}}

\newcommand{\rv}{\ensuremath{X}}

\newcommand{\population}{\ensuremath{\mathcal{P}}}
\newcommand{\aux}[1]{\ensuremath{aux(#1)}}
\newcommand{\dimension}{\ensuremath{d}}
\newcommand{\popexp}{\ensuremath{\boldsymbol{\mu}}}
\newcommand{\midpt}{\ensuremath{mid}}
\newcommand{\side}{\ensuremath{side}}

\newcommand{\popbn}{\ensuremath{BN(\population)}}
\newcommand{\outputnodes}{\ensuremath{O}}

\newcommand{\dataset}{\ensuremath{D}}
\newcommand{\numrecords}{\ensuremath{n}}
\newcommand{\datamarg}{\ensuremath{\overline{\boldsymbol{x}}}}
\newcommand{\datarecord}{\ensuremath{\mathbf{x}}}

\newcommand{\target}{\ensuremath{\mathbf{y}}}
\newcommand{\tlabel}{\ensuremath{\ell}}
\newcommand{\tin}{\ensuremath{\mathsf{IN}}}
\newcommand{\tout}{\ensuremath{\mathsf{OUT}}}
\newcommand{\tlabelparam}{\ensuremath{\tlabel_{\tin}}}

\newcommand{\pr}[1]{\ensuremath{Pr({#1})}}

\newcommand{\unifofset}[1]{\ensuremath{\mathcal{U}(#1)}}
\newcommand{\bindist}[2]{\ensuremath{\mathcal{B}(#1,\,#2)}}

\newcommand{\classicalthresh}{\ensuremath{T_{C}}}
\newcommand{\ipthresh}{\ensuremath{T_{IP}}}
\newcommand{\bayesthresh}{\ensuremath{T_{B}}}
\newcommand{\bayesratio}{\ensuremath{R}}
\newcommand{\freqratio}{\ensuremath{\Lambda}}

\newcommand{\flip}[1]{\ensuremath{flip(#1)}}
\newcommand{\networkprog}[1]{\ensuremath{network_{#1}}}
\newcommand{\observe}[1]{\ensuremath{observe(#1)}}
\usepackage{authblk}
\begin{document}

\title{A Bayesian Approach to Membership Inference for Statistical Release}

\author[1, 2]{Lisa Oakley}
\author[1]{Sam Stites}
\author[1]{Cameron Moy}
\author[1]{Steven Holtzen}
\author[1]{Alina Oprea}
\author[3]{Marco Gaboardi}

\affil[1]{\footnotesize Northeastern University, USA}
\affil[2]{\footnotesize Proof Trading, Inc, USA}
\affil[3]{\footnotesize Boston University, USA}

\date{\vspace{-5ex}}

\maketitle

\begin{abstract}
The membership inference problem for publicly released statistics from a private dataset is well-studied. When developing and formally analyzing attack strategies, however, the focus has been on attacks that model the population using only its marginals. In practice, these attacks can perform well on various populations, however most formal analysis is for populations that follow a product distribution. These strategies may fail to leverage useful information about the population that is important for understanding a realistic privacy threat.

In this work, we explore the impact of providing an attacker with additional information about the attribute dependency structure of the population, motivated by examples where multiple parties may have access to similarly structured data, for example the US Census and the IRS. To model this scenario, we reframe the membership inference problem with respect to a population represented as a Bayesian network (BN). We develop a framework based on Bayesian decision-making which can incorporate prior information about the population to launch more effective, specialized attacks. 

To evaluate our framework, we introduce a specific attack instantiation which computes the Bayesian posterior using a probabilistic program, and prove its equivalence to an optimal variant of the likelihood ratio test attack for two populations with strong attribute dependency. We implement our program in the Roulette probabilistic programming language and show experimentally that it outperforms the likelihood ratio test and inner product attacks on five commonly used BNs, where the population dependency structure is too complex for the existing attacks to be manually adapted.
\end{abstract}

\section{Introduction}
\label{sec:intro}
Membership inference attacks for statistical release seek to determine whether a given record is in a private dataset based on publicly available statistics. Popular attacks for the membership inference problem for statistical release tend to rely on assumptions about the population from which the private data were sampled \cite{SankararamanGenomic2009,homer_resolving_2008,inner_product_2015,survey_2017}. In particular, many well-known attacks such as the inner product and likelihood ratio test attacks use only the population marginals to compute their decisions. When the population follows a product distribution, there are theoretical guarantees on the effectiveness of these attacks, and using only the marginals as a proxy for the population allows for simple and efficient attacks. Even when the strong independence assumptions do not hold, these procedures can perform well on data with unknown or unavailable attribute dependency structure, however there is lacking theoretical investigation into why this is the case. In some cases, attacks that assume a product distribution can be manually adjusted to consider known structure of specific populations. For example, intentionally ignoring (clipping) equivalent attributes can correct for disproportionate signaling. While effective, there are limited cases with a known modification, and the process is hard to generalize. When analyzing the privacy of a system, it is important to understand the full capabilities of an attacker with some privileged information, rather than focusing only on the most convenient and efficient attacks.

In reality, it is rare to have a population with no dependency between attributes. For example, an individual's income is dependent on their age and occupation, and their occupation is dependent on their age and location. Not only are there obvious dependencies in the attributes of their data, there exists a potentially complex, multi-layered structural relationship between these attributes. 

It is also common that a curious party may have access to information about this structural relationship between attributes of the population. For example, when the US census releases public statistics about its respondents, it aims to protect the privacy of the individuals from all other parties, including other governmental departments. Many other departments might have data drawn from the same population, for example the IRS also has a sample of demographic information like age, state, number of children, and income over residents of the US. In this case, the census would want to understand the extent to which the IRS could use such a sample to inform their attack. In another alarming case, in 2025 Columbia University experienced a data breach affecting the personal information including demographic information, financial aid information, academic history, and Social Security numbers of applicants dating back over a decade and totaling over 870,000 affected individuals \cite{Columbia2025}. An attacker in this case could use this database of personal information to learn the attribute dependence structure of Ivy League college applicants in the US each year. The attacker could use this learned structure to inform a privacy attack on released admissions statistics (such as commonly released aggregate GPA, demographic, and geographical statistics) from other Ivy League universities who were not subject to this particular data breach. This type of advanced attack is not captured by analysis that considers an attacker who is only using the dataset marginals as attack inputs. 

One way to bring structure into statistical domains is by defining a Bayesian model of the population and using Bayesian decision-making \cite{murphy_bayesian_decision}. In this paper, we formalize this process for the membership inference problem to develop an attack evaluation framework that can incorporate the population's attribute dependency without manual analysis. We choose to model the population as a \emph{Bayesian network} (BN), which is a Bayesian model described as a directed, acyclic graph over random variables. BNs can flexibly describe many interesting distributions, and there are many methods for learning a BN from a sample dataset \cite{adnan_param_learning}.

Our primary contribution is an attack evaluation framework based on Bayesian decision-making which can be updated with a model of the population as a BN to create a specialized attack on the private data, without manual adjustments. Informally, our attack is a statistical test which, for threshold value $\bayesthresh$, the attacker guesses a target point $\target$ is in a given dataset $\dataset$ if

\begin{equation}
    \frac{\pr{\text{target is in dataset} \mid \text{observed quantities}}}{1-\pr{\text{target is in dataset} \mid \text{observed quantities}}} > \bayesthresh.
\end{equation}
In later sections, we define a probabilistic program for computing the posterior where the observed quantities include the description of the population as a BN and sampling procedure, and prove that our Bayesian attack finds an optimal clipping of the likelihood ratio test attack for two populations with high attribute dependency. Additionally, we show experimentally that an implementation of our Bayesian attack in the state-of-the-art Roulette \cite{moy2025roulette} probabilistic programming language is able to outperform the traditional attacks on private data generated by a set of well-known benchmarking BNs with complex attribute dependency structures that have no obvious clipping solution for the traditional attacks.

\section{A Bayesian View of the Membership Inference Problem}
\label{sec:background}
The membership inference problem for statistical release can be described simply as: assuming an attacker has access to a statistic over a private dataset and some information about the population from which the private data was drawn, the attacker must determine whether a target data point was drawn from the private data or from the population. Typical framing for this problem de-emphasizes potential population attribute dependency, often providing the attacker only the population marginals. We propose a reframing that provides the attacker more information about the population. In particular, we describe the population as a \emph{Bayesian network}, and provide some auxiliary information to the attacker. In our formulation, we focus on the membership inference problem for distributions with binary (or one-hot-encoded discrete) attributes. This formulation can be immediately extended to populations with continuous attributes by discretizing (bucketing) the feature space, or by using a Bayesian network over continuous features.

\subsection{The Membership Inference Problem for Statistical Release}
The membership inference problem for statistical release can be described as:

\begin{enumerate}
    \item Let population $\population:\Omega\rightarrow [0,1]$ over $\Omega=\{0,1\}^\dimension$ be a discrete data distribution, and let target label $\tlabel\in\{\tin,\tout\}$.
    \item Private dataset $\dataset=\{\mathbf{x}^{(1)},\dots,\mathbf{x}^{(\numrecords)}\}$ consists of $\numrecords$ samples (records) drawn i.i.d. from $\population$.
    \item The sample means $\datamarg=\frac{1}{\numrecords}\sum_{i=1}^{\numrecords} \mathbf{x}^{(i)}$ of $\dataset$ are publicly released.
    \item Attacker is given $\numrecords$, and some information $\aux{\population}$ about the population.
    \item Attacker is given target point $\target$ such that 
    \begin{equation}
        \target=\begin{cases}
            \target\sim\unifofset{\dataset},& \text{if } \tlabel = \tin\\
            \target\sim\population,  & \text{if } \tlabel = \tout
        \end{cases}
    \end{equation}
    with $\unifofset{\dataset}$, the uniform distribution over the records of $\dataset$.
    \item Attacker must decide whether $\tlabel$ is \tin{} or \tout{}.
\end{enumerate}

This formulation does not assume the population has attribute independence. Threat models under this definition are instantiations of $\aux{\population}$. We discuss some example threat models in Sec. \ref{apdx:threatmodels}. We use boldface variables to indicate a vector, $\mathbf{a}_j$ is the $j$th element of vector $\mathbf{a}$, and $\mathbf{x}^{(i)}_j$ represents the $j$th element of the $i$th record of database $\dataset$.

\subsection{Population as a Bayesian Network}

Our formulation of the membership inference problem describes the population as a discrete distribution. To allow for rich graphical descriptions of the structure of the distribution, we use Bayesian networks to model this discrete distribution. There are other graphical structures, for example this causal graph of the German credit database \cite{machado2025sequential}, that could also be potential population models, but we focus on Bayesian networks here for their generality, availability of examples, and standardized format.

A \emph{Bayesian Network (BN)} is a directed, acyclic graph (DAG) where nodes are random variables and edges are dependencies. It is used to succinctly capture the structure of dependency of the random variables in a distribution, with each node containing a joint probability table over itself and its parents. The power of modeling a distribution as a BN comes from the ability to define the whole dependency structure of the ancestors of a random variable with respect to only its parents. We denote the BN which describes a population \population{} as \popbn{}.

For example, we consider a simple BN from Korb and Nicholson \cite{bn_cancer}, as described in Fig. \ref{fig:bn_cancer}, which models probabilities of cancer symptoms dependent on environmental factors. This BN consists of five binary random variables (nodes): Pollution (P), Smoking (S), Cancer (C), X-ray (X) and Dyspnoea (D). In this example, the probability of having cancer (C) is dependent on pollution (P) and smoking (S), and the probability that the patient has a positive x-ray (X) or shortness of breath (D) are dependent on their shared parent (C).

\begin{figure*}
    \centering
    \begin{subfigure}{.4\linewidth}
        \centering
        \includegraphics[width=.9\linewidth]{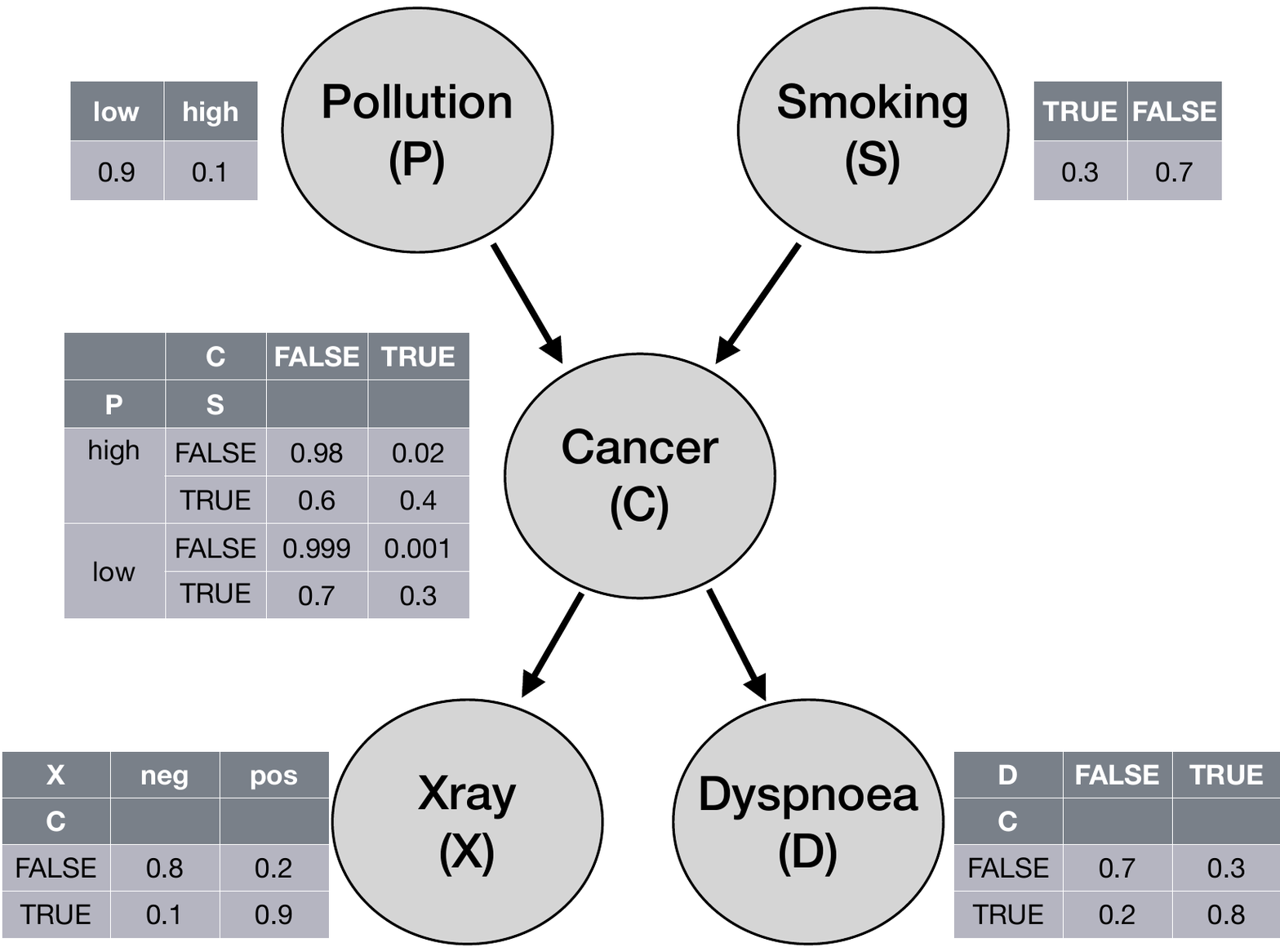}
        \caption{Bayesian network example from \cite{bn_cancer}. Conditional probability values are modified for the clarity of this example.}
        \label{fig:bn_cancer}
    \end{subfigure}
    \qquad
    \begin{subfigure}{.2\linewidth}
        \centering
        \includegraphics[width=1\linewidth]{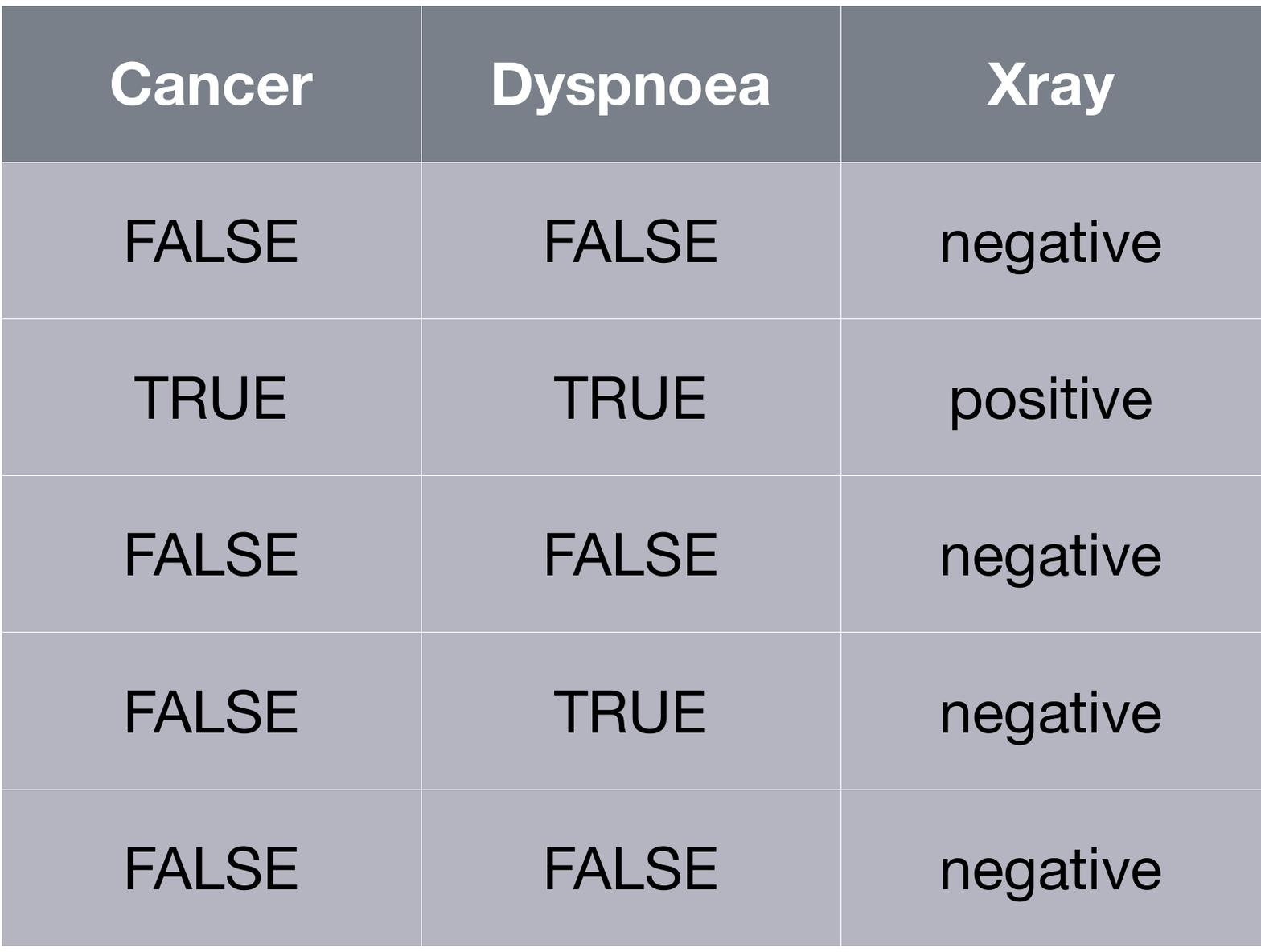}
        \caption{Sample from Cancer Bayesian network example with output nodes $\outputnodes=\{C,X,D\}$ and $\numrecords=5$ samples.}
        \label{fig:cancer_sample}
    \end{subfigure}
    \caption{Simple example of a BN (left), and a sample drawn from this BN (right).}
\end{figure*}

\begin{figure}
    \centering
    \includegraphics[width=.4\linewidth, trim=0 6cm 0 6cm, clip]{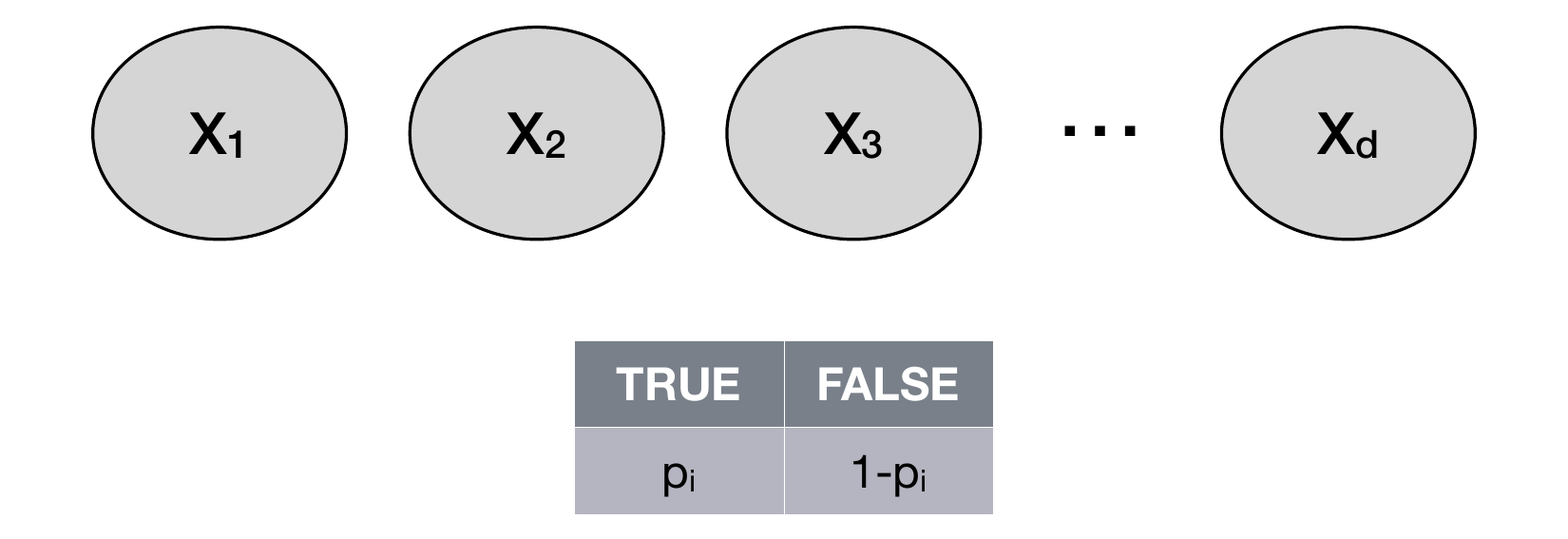}
    \caption{Binary product distribution as a Bayesian network where $\mathbf{p}_i=\popexp_i$ represents the probability that $X_i$ is true.}
    \label{fig:bn_product}
\end{figure}

\subsubsection{Describing a Product Distribution as a Bayesian Network}\label{apdx:prodbn}
The typical membership inference problem (and existing attacks) assumes the population follows a product distribution. This is a special case of our formulation, with the Bayesian network described in Fig. \ref{fig:bn_product}. Therefore, our formulation is a generalization of the typical membership inference problem formulation.

\subsubsection{Sampling a Private Dataset from a Bayesian Network}
The membership inference problem is defined for private datasets of $n$ records drawn i.i.d. from a population with $d$ attributes. 
When describing the population as a Bayesian network, we define a set of \emph{output nodes}, \outputnodes, whose values correspond to these attributes. The output nodes are chosen based on the semantics of the Bayesian network as part of the definition of the population, and are consistent across attacks. In Fig. \ref{fig:bn_cancer}, we can reasonably define the output nodes as $\{X,D\}$, meaning that the private dataset only contains information about the symptoms the patient is experiencing ($\dimension=2$), as $\{C,X,D\}$ ($\dimension=3$) including the patient status, or as $\{P, S, C, X, D\}$ ($\dimension=5$) the entire set of random variables. Fig. \ref{fig:cancer_sample} is an example of a sampled private dataset with $\numrecords=5$ records for $\outputnodes=\{C,X,D\}$.

\subsection{Attacker Information (Threat Models)}\label{apdx:threatmodels}
Now that we have defined the population \population{} in terms of a Bayesian network, we can discuss how to instantiate the auxiliary information the attacker receives ($\aux{\population}$) to define different threat models. In particular, we explain the strong, weak, and weakest attacker threat models. In this paper, we will focus on the strong attacker scenario for our theoretical analysis. In Sec. \ref{sec:bayesonbn}, we provide experimental results for all three threat models.

\subsubsection{Full Information (Strong Attacker)}
The strongest attacker threat model in the case of a population modeled as a Bayesian network provides the full structure and joint probability tables of nodes in the Bayesian network. Formally, $\aux{\population}=\popbn$. This provides the attacker with the maximum amount of information about the Bayesian network. 

\subsubsection{Public Auxiliary Dataset and Structure (Weak Attacker)}
In the case of the weak attacker threat model, the attacker has access to only the \emph{structure} of the Bayesian network (i.e., nodes and edges), but not to the joint probability tables of each node. Additionally, this attacker has access to a public auxiliary dataset drawn from the same population as the private dataset. The attacker can use state-of-the-art methods to learn the joint probability tables of the Bayesian network to get an approximation of the true population \cite{adnan_param_learning}. The attacker can then use this approximate Bayesian network to implement the same attack as the strong attacker.

\subsubsection{Auxiliary Dataset and No Structure (Weakest Attacker)}
The weakest attacker gets only the public auxiliary dataset drawn from the same distribution of the private dataset. This attacker will have to learn both the structure of the Bayesian network, and the joint probability tables. Learning the structure of the Bayesian network is a harder problem than just learning the joint probabilities, but there are methods that exist to do this type of learning as well \cite{kitson2023structurelearning}.

\subsection{Measuring Attack Utility}\label{sec:utility}

In our experimental evaluation, we evaluate the utility of each attack with respect to its \textit{Receiver Operating Characteristic (ROC) curve}, and specifically the \textit{area under this curve (AUC)}. The ROC curve is a plot of the true positive rate (TPR), when attacker guesses the target is in the dataset in the case where $\tlabel=\tin$, over the false positive rate (FPR), when attacker guesses the target is in the dataset in the case where $\tlabel=\tout$. 

Intuitively, high AUC means that a target point drawn from the private database will have a higher probability of being labelled $\tin$ than a target point drawn directly from the population. An AUC of 1 means the attacker correctly classifies all target points. Therefore, the mean AUC over various private databases drawn from the same population is a good metric for determining attack utility for specific populations or Bayesian networks.

\section{Population Marginal Attacks under Attribute Dependence}
\label{sec:classical}
Before introducing our Bayesian approach, we highlight the limitations of existing attacks for membership inference on statistical release in our context. We categorize the likelihood ratio test and inner product attacks under the umbrella of ``population marginal attacks,'' referring to their reliance on the population marginals (or a proxy thereof) as their population model. We demonstrate how these attacks can fail under strong population attribute dependence.

\subsection{Likelihood Ratio Test Attack with a Product Assumption}\label{sec:classical_attack}
The likelihood ratio test attack \cite{SankararamanGenomic2009,homer_resolving_2008} uses hypothesis testing with $H_0$ (null): \target{} was sampled from the population (i.e. \target{} is \tout{}) and $H_1$: \target{} was chosen from $\dataset$ (i.e. \target{} is \tin{}) to decide whether a sample is in or out of the dataset. In particular, it finds the likelihood ratio
\begin{equation}
    \freqratio = \frac{\Pr_{\target'\sim \unifofset{\dataset}}(\target'=\target)}{\Pr_{\target'\sim \population}(\target'=\target)},
\end{equation}
and chooses threshold $\classicalthresh$ to makes the decision:
    $\text{\emph{attacker guesses} }\tlabel = \tin \iff \Lambda > \classicalthresh.$

By the Neyman-Pearson Lemma, this is the optimal statistical test for $H_0$ when the likelihoods are computed exactly. However, in  implementation, an approximate ratio is often computed using only the population marginals (assuming attribute independence):
\begin{equation}
    \freqratio = 
    \frac{\prod_{j=1}^{\dimension} \{\text{if } \target_j = 1 \text{ then } \datamarg_j \text{ else } 1-\datamarg_j\}}
    {\prod_{j=1}^{\dimension} \{\text{if } \target_j = 1 \text{ then } \popexp_j \text{ else } 1-\popexp_j\}}. \label{eqn:freq_test}
\end{equation}

\subsection{Inner Product Attack}
\label{sec:inner_product_attack}
Another popular attack, the \emph{inner product attack} \cite{inner_product_2015}, computes how inclusion of the target point changes the released statistic to make the decision:
$\text{\emph{attacker guesses} }\tlabel = \tin \iff \langle \datamarg - \popexp, \target \rangle > \ipthresh$ 
where $\ipthresh$ is some threshold, $\datamarg$ are the sample/dataset means, $\popexp$ are the population marginals, $\target{}$ is the target point and $\langle \mathbf{a},\mathbf{b} \rangle$ is the inner product of vectors $\mathbf{a}$ and $\mathbf{b}$. This test is succinct, but also limits population evidence to solely the population marginals. 

\begin{figure}
    \centering
    \begin{subfigure}{.3\linewidth}
        \centering
        \includegraphics[width=\linewidth]{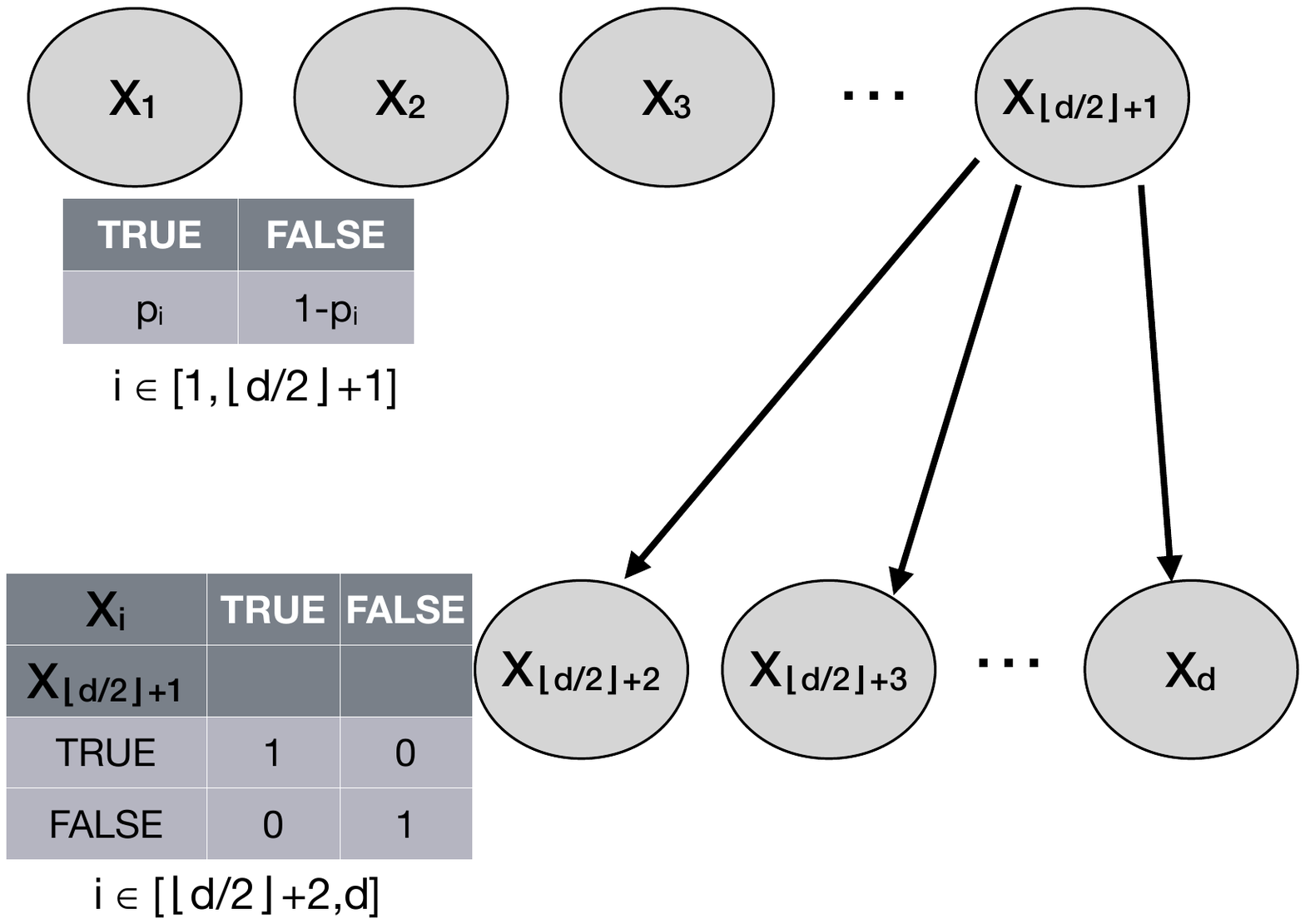}
        \caption{BN for half repeated data. Here the output set $\outputnodes=\{\rv_1,\dots,\rv_d\}$.}
    \label{fig:bn_half}
    \end{subfigure}
    \qquad
    \begin{subfigure}{.5\linewidth}
        \centering
        \includegraphics[width=\linewidth]{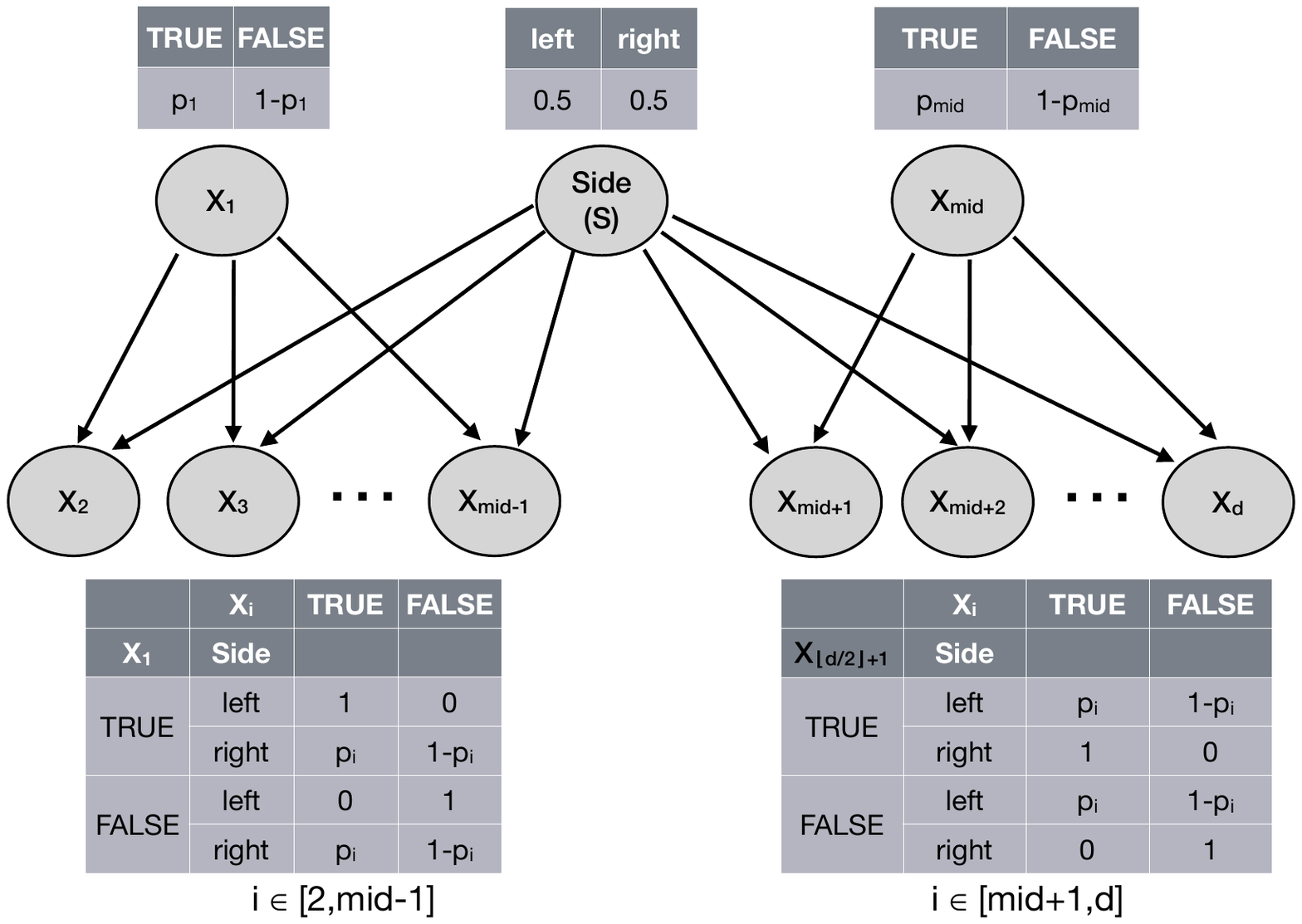}
        \caption{BN for left/right repeated data for even $\numrecords$. We set $\midpt=\lfloor\frac{\dimension}{2}\rfloor+1$ and $\outputnodes=\{\rv_1,\dots,\rv_d\}$.}
        \label{fig:bn_side}
    \end{subfigure}
    \caption{Toy Bayesian networks with strong attribute dependencies.}    
\end{figure}

\subsection{Performance under Strong Attribute Dependence}
Though the likelihood ratio test attack is optimal on product-distributed populations, and the likelihood ratio test and inner product attacks both boast strong performance with respect to certain metrics \cite{SankararamanGenomic2009,homer_resolving_2008,inner_product_2015,survey_2017}, there are simple attribute dependency structures on which they struggle. In two such cases, the dependency structure of the population is simple enough that there is a clipping which neutralizes the effects of this dependency. However, in many cases the attribute dependency structure is too complicated for simple population-specific fixes. We describe these two cases and their associated clippings in the following sections.

\subsubsection{Half Repeated Population}\label{sec:freq_attack_half}
In the \emph{half repeated population} (described in Fig. \ref{fig:bn_half}), the first half of the attributes are independent, but the other half are repeated. Intuitively, population marginal attacks amplify the signal of the $(\lfloor\dimension/2\rfloor+1)$th attribute, overwhelming the meaningful information in the first half of the attributes. The population marginal attackers can correct for this by observing the simple dependency and clipping the repeated attributes. Formally, the likelihood ratio test attacker computes
\begin{equation}
    \freqratio_{half}=\frac
        {\prod_{j=1}^{\lfloor\dimension/2\rfloor+1}\left\{\text{if } \target_j=1 \text{ then }\datamarg_j \text{ else }1-\datamarg_j\right\}}
        {\prod_{j=1}^{\lfloor\dimension/2\rfloor+1}\left\{\text{if } \target_j=1 \text{ then }\popexp_j \text{ else }1-\popexp_j\right\}},
        \label{eqn:freq_test_half}
\end{equation}
and the inner product attacker computes the inner product for the first $(\lfloor\dimension/2\rfloor+1)$ attributes for the half repeated population.

\subsubsection{Left/Right (l/r) Repeated Population}\label{sec:freq_attack_side}
In the slightly more complicated \emph{left/right (l/r) repeated population} (described in Fig. \ref{fig:bn_side}), half of the attributes are again repeated, but we introduce an additional coin flip to determine which side of the attributes are independent. The attacker can still correct for signal amplification here by clipping repeated attributes, however they must first determine the correlated side from the released statistics. Formally the likelihood ratio test attacker computes for $\side=right$
\begin{equation}
    \freqratio_{\midpt,right}=\frac
        {\prod_{j=1}^{\midpt}\left\{\text{if } \target_j=1 \text{ then }\datamarg_j \text{ else }1-\datamarg_j\right\}}
        {\prod_{j=1}^{\midpt}\left\{\text{if } \target_j=1 \text{ then }\popexp_j \text{ else }1-\popexp_j\right\}}
        \label{eqn:freq_test_side_right}
\end{equation}
and for $\side=left$
\begin{equation}
    \freqratio_{\midpt,left}=\frac
        {\prod_{j=\midpt-1}^{\dimension}\left\{\text{if } \target_j=1 \text{ then }\datamarg_j \text{ else }1-\datamarg_j\right\}}
        {\prod_{j=\midpt-1}^{\dimension}\left\{\text{if } \target_j=1 \text{ then }\popexp_j \text{ else }1-\popexp_j\right\}}.
        \label{eqn:freq_test_side_left}
\end{equation}
for the l/r repeated population.

Choosing the repeated side incorrectly causes this attack to have worse performance than ignoring the repetition altogether. In Fig. \ref{fig:corr_fails}, we demonstrate the impact of attribute correlation on the original and clipped attacks.
\begin{figure}
    \centering
    \includegraphics[width=.5\linewidth]{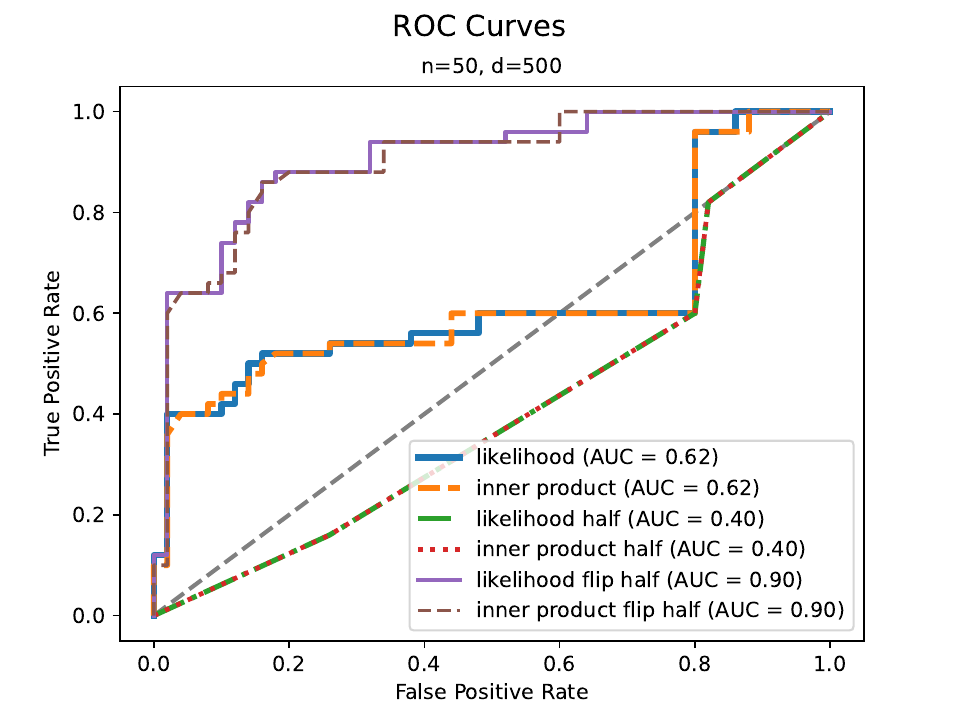}
    \caption{An example of ROC curves for population marginal attacks on left/right repeated populations (illustrated in Fig \ref{fig:bn_side}). The basic attacks (likelihood, inner product) perform poorly on data with this correlation structure. When the attacker knows which side is correlated (likelihood half, inner product half), they are able to improve the attack performance significantly by clipping repeated attributes. When the attacker guesses the correlated side incorrectly (likelihood flip half, inner product flip half), they perform worse than a coin flip because they clip the useful information and leave only the repeated attributes. }
    \label{fig:corr_fails}
\end{figure}
\subsubsection{More Complicated Dependency Structure}
In populations with only independent or repeated attributes and simple structure, clipping is effective for recovering performance of population marginal attacks. As the structure gets even slightly more complicated, however, fixes become less obvious. Consider a population where $k$ random attributes are chosen to take the same value. How does the attacker decide which attributes to ignore? For the Cancer example from Sec. \ref{sec:background}, any clipping would remove useful signal. In the next section, we describe our Bayesian decision-making attack, which is defined generally for any population expressed as a Bayesian network, and demonstrate its strength at incorporating attribute dependency structure without manual adjustments.

\section{Our Bayesian Attack Framework}
\label{sec:bayes}
Our proposed attack framework is based on known Bayesian decision-making procedures \cite{murphy_bayesian_decision}. In this case, we are trying to estimate the probability that $\tlabel=\tin$, which we will represent using the parameter $\tlabelparam$. We compute the posterior using auxiliary information given to the attacker, including the Bayesian network of the population, the dataset marginals, and the target point. More specifically, though still informally, we compute a Bayesian decision ratio

\begin{equation}
    \bayesratio = \frac
    {\pr{\tlabelparam \mid \popbn,\,\datamarg,\,\target}}
    {1-\pr{\tlabelparam \mid \popbn,\,\datamarg,\,\target}}.
    \label{eqn:bayestest}
\end{equation}
Our attack procedure for some thresholding value, $\bayesthresh$, is therefore 
\begin{equation}
    \text{attacker guesses }\tlabel = \tin \iff R > \bayesthresh.
\end{equation} 

In the case of the weak and weakest attack models as described in Sec. \ref{apdx:threatmodels}, the attacker must first approximate this Bayesian network from a proxy dataset drawn from the population. 

\subsection{Defining the Posterior as a Program}
The quantity $\pr{\tlabelparam \mid \popbn,\,\datamarg,\,\target}$ represents the probability that $\tlabel=\tin$, conditioned on the marginals, $\datamarg$, being equivalent to the marginals of a dataset where one arbitrary record is the target point \target{}, and the other $\numrecords-1$ records are sampled i.i.d. from the population. The denominator of eqn. (\ref{eqn:bayestest}) is the probability that $\tlabel=\tin$, given the marginals $\datamarg$ are equal to those of a dataset where all $\numrecords$ are drawn i.i.d. from the population (without assuming the target is in the data).

We formally define this posterior as a probabilistic program. We define $\networkprog{\popbn}$, a program that models the population Bayesian network $\popbn$.  In Appendix \ref{apdx:network_ppl}, we provide an example of the encoded Cancer Bayesian network from Fig. \ref{fig:bn_cancer}, written in the Roulette probabilistic programming language (PPL). We also assume we have two primitives: $\flip{p}$, which represents a Bernoulli distribution with parameter $p$, and $\observe{\cdot}$, a standard PPL construct that calculates a conditional probability. Concretely, for the program $\theta\gets\flip{1/2};$ $\observe{\phi_1}\dots\observe{\phi_n};$ $\text{return }\theta$, the final output will be $\pr{\theta\mid\phi_1\dots\phi_n}$. Here, $\theta$ is initialized to the prior $\flip{1/2}$, and the $\observe{\cdot}$ statements condition the output. The semantics of the PPL handles computing the posterior. 

The process of computing the posterior described by a probabilistic program is defined in the literature as \emph{probabilistic inference} and is generally intractable \cite{holtzen2020scaling}. Similarly to SAT solving, there are techniques to solve these problems exactly and efficiently in some cases using Bayesian methods. For example, in Roulette, posterior inference is performed by computing the normalizing constant exactly via knowledge compilation. Importantly, this does not involve any sampling and instead exactly computes the posterior. PPLs that perform approximate inference, say using Markov-chain Monte Carlo, approximate the posterior. The surface syntax using $\observe{\cdot}$, however, remains the same. 

The probabilistic program to compute our desired posterior is defined in Algorithm \ref{alg:posterior} and works as follows. On line (\ref{line:prior}), we initialize the prior $\tlabelparam$ as a fair coin flip. On line (\ref{line:network}), we initialize $\numrecords$ data records to be independent instances of the Bayesian network for the population. On lines (\ref{line:target}) and (\ref{line:marginals}) we condition on the sums of the attributes of the instances of the Bayesian networks to be equal to the given marginals of the private dataset. On line (\ref{line:target}), we choose an arbitrary record to be the target interval if $\tlabelparam$. In this way, we indicate that the target point is in the private dataset.  On line (\ref{line:posterior}) we return the posterior distribution, $\pr{\tlabelparam \mid \popbn,\,\datamarg,\,\target}$, as desired.

\begin{algorithm}
    \caption{\label{alg:posterior} Compute Posterior $\pr{\tlabelparam \mid \popbn,\,\datamarg,\,\target}$} 
    \begin{algorithmic}
    \Require Bayesian network \popbn, marginals \datamarg{} of private dataset, target point \target
    \Procedure{}{\popbn, \datamarg, \target}
        \State $\tlabelparam \gets \flip{0.5}$ \label{line:prior}\Comment{initialize prior}
        \For {$i$ \textbf{in }$[\numrecords]$}
            \State $\datarecord^{(i)} \gets \networkprog{\popbn}$ \label{line:network}\Comment{set networks}
        \EndFor
        \If {$\tlabelparam$}
            \State $\observe{\target + \sum_{i=1}^{\numrecords-1}\datarecord^{(i)}=\datamarg\cdot n}$  \label{line:target}\Comment{one record is $\target$, the rest from $\population$}
        \Else
            \State $\observe{\sum_{i=1}^{\numrecords}\datarecord^{(i)}=\datamarg\cdot n}$  \label{line:marginals}\Comment{all records from $\population$}
        \EndIf
        
        \State \Return $\tlabelparam$ \label{line:posterior}\Comment{return posterior}
    \EndProcedure
\end{algorithmic}
\end{algorithm}

\section{Correctness of the Bayesian Framework}
\label{sec:bayesoncorr}
The likelihood ratio test (LRT) attack with a product assumption from Sect. \ref{sec:classical_attack} is the strongest statistical test for product-distributed populations by the Neyman-Pearson lemma \cite{SankararamanGenomic2009,homer_resolving_2008,survey_2017}. Simple clippings preserve optimality for half and l/r repeated populations. We present three main results showing our proposed Bayesian attack is equivalent to the (clipped) LRT attack on these populations  in the case of the strong threat model (using two lemmas from Appendix \ref{apdx:lemmas}).

The key takeaway here is that, while we had to define bespoke versions of the population marginal attacks to adapt to different attribute dependency, the Bayesian approach is more modular. There, the only change between attacks is defining the network structure to match that of the Bayesian network that describes the population. This means that the Bayesian framework is far more adaptive than the traditional attacks, and does not require manual modification.

\begin{theorem}[Product Distributed Correctness]\label{thm:product}
    Assuming a population that follows a product distribution, the Bayesian framework using Algo. \ref{alg:posterior} to compute the posterior is equivalent to the likelihood ratio test attack from Sect. \ref{sec:classical_attack}.
\end{theorem}
\begin{proof}
    Because the Bayesian attacker guesses $\tlabel=\tin\iff \bayesratio>\bayesthresh$, and the Likelihood attacker guesses $\tlabel=\tin\iff \freqratio>\classicalthresh$, if we set $\bayesthresh=\classicalthresh$, the attacks are equivalent by Lemma \ref{lem:ratioeq} ($\bayesratio=\freqratio$). 
\end{proof}

\begin{theorem}[Half Distributed Correctness]\label{thm:half}
    Assuming a population that follows a half repeated distribution described in Fig. \ref{fig:bn_half}, the Bayesian framework using Algo. \ref{alg:posterior} to compute the posterior is equivalent to the clipped likelihood ratio test attack from Sec. \ref{sec:freq_attack_half}.
\end{theorem}
\begin{proof}
    Because the Bayesian attacker guesses $\tlabel=\tin\iff \bayesratio>\bayesthresh$, and the Likelihood attacker guesses $\tlabel=\tin\iff \freqratio>\classicalthresh$. If we set $\bayesthresh=\classicalthresh$, $\side=right$, and $\midpt=\lfloor\dimension/2\rfloor$, by definition $\freqratio=\freqratio_{\midpt,\side}$, and the attacks are equivalent by Lemma \ref{lem:partialratioeq} ($\bayesratio_{half}=\freqratio_{\midpt,\side}$). 
\end{proof}

\begin{theorem}[Left/Right Distributed Correctness]\label{thm:side}
    Assuming a population that follows a left/right repeated described in Fig. \ref{fig:bn_side}, the Bayesian framework using Algo. \ref{alg:posterior} to compute the posterior is equivalent to the likelihood ratio test attack from Sect. \ref{sec:freq_attack_side}.
\end{theorem}
\begin{proof}
    Because the Bayesian attacker guesses $\tlabel=\tin\iff \bayesratio>\bayesthresh$, and the Likelihood attacker guesses $\tlabel=\tin\iff \freqratio_{\midpt,\side}>\classicalthresh$, if we set $\bayesthresh=\classicalthresh$, the attacks are equivalent by Lemma \ref{lem:partialratioeq} ($\bayesratio=\freqratio_{\midpt,\side}$).
\end{proof}

\section{Evaluating the Bayesian Framework using Probabilistic Programming}
\label{sec:bayesonbn}
\begin{table}
\centering
\begin{tabular}{lllll}
\toprule
BN (\dimension) & Attack  & AUC & Std. \\
\midrule
half rep. (25) & \textbf{Bayes} & \textbf{0.743} & 0.103 \\
half rep. (25) & LRT & 0.696 & 0.107 \\
half rep. (25) & IP & 0.683 & 0.101 \\
\hline
cancer\cite{bn_cancer} (10) & \textbf{Bayes} & \textbf{0.744} & 0.106 \\
cancer\cite{bn_cancer} (10) & LRT & 0.710 & 0.111 \\
cancer\cite{bn_cancer} (10) & IP & 0.688 & 0.111 \\
\hline
quake\cite{bn_earthquake} (10) & \textbf{Bayes} & \textbf{0.594} & 0.082 \\
quake\cite{bn_earthquake} (10) & LRT & 0.534 & 0.100 \\
quake\cite{bn_earthquake} (10) & IP & 0.533 & 0.101 \\
\bottomrule
\end{tabular}
\quad
\begin{tabular}{lllll}
\toprule
BN (\dimension) & Attack  & AUC & Std. \\
\midrule
asia \cite{bn_asia} (16) & \textbf{Bayes} & \textbf{0.763} & 0.123 \\
asia \cite{bn_asia} (16) & LRT & 0.652 & 0.150 \\
asia \cite{bn_asia} (16) & IP & 0.628 & 0.154 \\
\hline
survey\cite{bn_survey} (14) & \textbf{Bayes} & 0.837 & 0.095 \\
survey\cite{bn_survey} (14) & LRT & \textbf{0.839} & 0.098 \\
survey\cite{bn_survey} (14) & IP & 0.815 & 0.102 \\
\hline
sachs\cite{bn_sachs} (15) & \textbf{Bayes} & \textbf{0.906} & 0.066 \\
sachs\cite{bn_sachs} (15) & LRT & 0.851 & 0.095 \\
sachs\cite{bn_sachs} (15) & IP & 0.824 & 0.107 \\
\bottomrule
\end{tabular}
\caption{Comparing our Bayesian (Bayes) attack AUC with the likelihood ratio test (LRT) and inner product (IP) attacks under the strong threat model. For the Sachs Bayesian network (BN), we define the output nodes to be the nodes in the longest path of the BN. In all other cases, the output set is the set of all nodes.}
\label{tab:accuracy}
\end{table}

The theoretical results from Sect. \ref{sec:bayesoncorr} indicate that the Bayesian approach finds the optimal clipping in some cases. Now, we present experimental results to demonstrate the utility of our method when there is no known fix for population marginal attacks. We evaluate our overall performance for the strong, weak, and weakest attacker threat models by comparing AUC to the likelihood ratio test (LRT) and inner product (IP) attacks on various Bayesian networks (BNs). We further vary the output nodes of one BN to understand attack performance over attribute connectedness and find that incorporating attribute structure into the attack becomes more valuable as the dimension and attribute dependency increases. We also investigate the impact of increasing the number of records for selected examples and provide scalability analysis of this solution on a variety of benchmark populations using the state-of-the-art Roulette probabilistic programming language for exact probabilistic inference.

Our experiments are written in Python and run on a 16 core AMD EPYC with 64GB of RAM. We compute the exact Bayesian posterior (Alg. \ref{alg:posterior}) using the Roulette probabilistic programming language \cite{moy2025roulette}. We pull standard benchmarking BNs from the bnlearn repo \cite{bnlearn} sampled using Pgmpy \cite{pgmpy}, and average AUCs over 40 trials (unique private datasets drawn i.i.d.), with 20 in and 20 out targets. We choose $\numrecords<\dimension{}$ to ensure enough signal for attack effectiveness.

\subsection{Attack Performance on Selected Bayesian Networks (Strong Threat Model)}\label{sec:performance}
Our main objective is to understand how incorporating a Bayesian model of the population into the attack framework affects performance in practice. We showed theoretically in Sec. \ref{sec:bayesoncorr} that our Bayesian attack recovers the clipped LRT attack on the half and left/right repeated populations. In Table \ref{tab:accuracy}, we confirm experimentally that the Bayesian attack outperforms the standard LRT and IP attacks on the half repeated population for the strong threat model. 

We also want to look at more realistic populations with more complicated dependency structures that go beyond independent or repeated attributes. In our experiments, we consider a selection of standard benchmarking BNs, and compare the AUC of our Bayesian posterior attack to the population marginal attacks. We see in Table \ref{tab:accuracy} that the Bayesian approach implemented in Roulette performs as well as or better than the population marginal attacks. Even with as little as 15 correlated parameters as in the Sachs BN, we see a mean attack AUC that is 5\% higher than the LRT attack, and 8\% higher than the IP attack.

\begin{figure*}
    \centering
    \begin{subfigure}{.4\linewidth}
        \centering
        \scalebox{0.9}{\begin{tabular}{llllll}
\toprule
BN (\dimension) & $m$ & Attack  & AUC & Std. \\
\midrule
asia \cite{bn_asia} (16) & 10 & \textbf{Bayes} & \textbf{0.773} & 0.108 \\
asia \cite{bn_asia} (16) & 10 & LRT       & 0.699 & 0.133 \\
asia \cite{bn_asia} (16) & 10 & IP     & 0.685 & 0.136 \\
\midrule
asia \cite{bn_asia} (16) & 50 & \textbf{Bayes}  & \textbf{0.770} & 0.115 \\
asia \cite{bn_asia} (16) & 50 & LRT        & 0.691 & 0.137 \\
asia \cite{bn_asia} (16) & 50 & IP      & 0.678 & 0.137 \\
\midrule
asia \cite{bn_asia} (16) & 100 & \textbf{Bayes}  & \textbf{0.773} & 0.112 \\
asia \cite{bn_asia} (16) & 100 & LRT        & 0.691 & 0.138 \\
asia \cite{bn_asia} (16) & 100 & IP      & 0.676 & 0.142 \\
\bottomrule
\end{tabular}}
        \caption{Public Auxiliary Dataset and Structure (Weak Attacker)}
        \label{fig:asia_weak}
    \end{subfigure}
    \qquad
    \hspace{2em}
    \begin{subfigure}{.4\linewidth}
        \centering
\scalebox{0.9}{\begin{tabular}{lllll}
\toprule
BN (\dimension) & $m$ & Attack  & AUC & Std. \\
\midrule
asia \cite{bn_asia} (16) & 10 & \textbf{Bayes}  & \textbf{0.770} & 0.108 \\
asia \cite{bn_asia} (16) & 10 & LRT  & 0.689 & 0.137 \\
asia \cite{bn_asia} (16) & 10 & IP  & 0.674 & 0.139 \\
\midrule
asia \cite{bn_asia} (16) & 50 & \textbf{Bayes}  & \textbf{0.768} & 0.110 \\
asia \cite{bn_asia} (16) & 50 & LRT  & 0.685 & 0.133 \\
asia \cite{bn_asia} (16) & 50 & IP  & 0.669 & 0.136 \\
\midrule
asia \cite{bn_asia} (16) & 100 & \textbf{Bayes}  & \textbf{0.768} & 0.110 \\
asia \cite{bn_asia} (16) & 100 & LRT  & 0.683 & 0.132 \\
asia \cite{bn_asia} (16) & 100 & IP  & 0.669 & 0.135 \\
\bottomrule
\end{tabular}}
        \caption{Auxiliary Dataset and No Structure (Weakest Attacker)}
        \label{fig:asia_weakest}
    \end{subfigure}
    \caption{Comparing our Bayesian (Bayes) attack AUC with the likelihood ratio test (LRT) and inner product (IP) attacks under the weak and weakest threat models on the Asia benchmarking Bayesian network with proxy datasets of $m$ records.}
    \label{fig:weak_weakest}
\end{figure*}

\subsection{Attack Performance Under Weaker Threat Models}\label{sec:weaker}
To evaluate the weaker threat models, we modify the auxiliary information available to the attackers. For the weak threat model (public auxiliary dataset and structure), the attackers get access to the dependency structure of the population Bayesian network and a proxy dataset of $m$ records, drawn from the population. To evaluate the Bayesian attacker, we use the Maximum Likelihood Estimator from the PGMPY package \cite{pgmpy} to learn the joint probabilities of each node in the Bayesian network from the proxy dataset. This learned Bayesian network is then used as the input to Algorithm \ref{alg:posterior}. For the weakest threat model (auxiliary dataset and no structure), the attackers only get access to a proxy dataset of $m$ records drawn from the population. To evaluate the Bayesian attacker here, we use the PC algorithm from the PGMPY package \cite{pgmpy} to learn both the structure and joint probabilities from the proxy data, and again use this as input to Algorithm \ref{alg:posterior}. Because the population marginal attacks are fundamentally unable to leverage any information about the dependency structure of the features of the dataset, in both the weak and weakest attacker cases we approximate the dataset means by computing the marginals of the proxy dataset.

In Fig. \ref{fig:weak_weakest} we see the Bayesian attacker maintains a similar advantage under the weak threat model (approximating the joint probabilities), and in the weakest threat model (learning the structure of the Bayesian network and approximating the joint probabilities) for proxy datasets of size $10$, $50$, and $100$. We include results for the Asia benchmarking Bayesian network \cite{bn_asia} in Fig. \ref{fig:weak_weakest}, and provide results for additional benchmarking Bayesian networks in Appendix \ref{apdx:weak_weakest}. This is evidence that this attack remains effective for an attacker who has similarly distributed data to the private dataset.

\begin{figure*}
    \centering
    \begin{subfigure}{.2\linewidth}
        \centering
        \includegraphics[width=\linewidth]{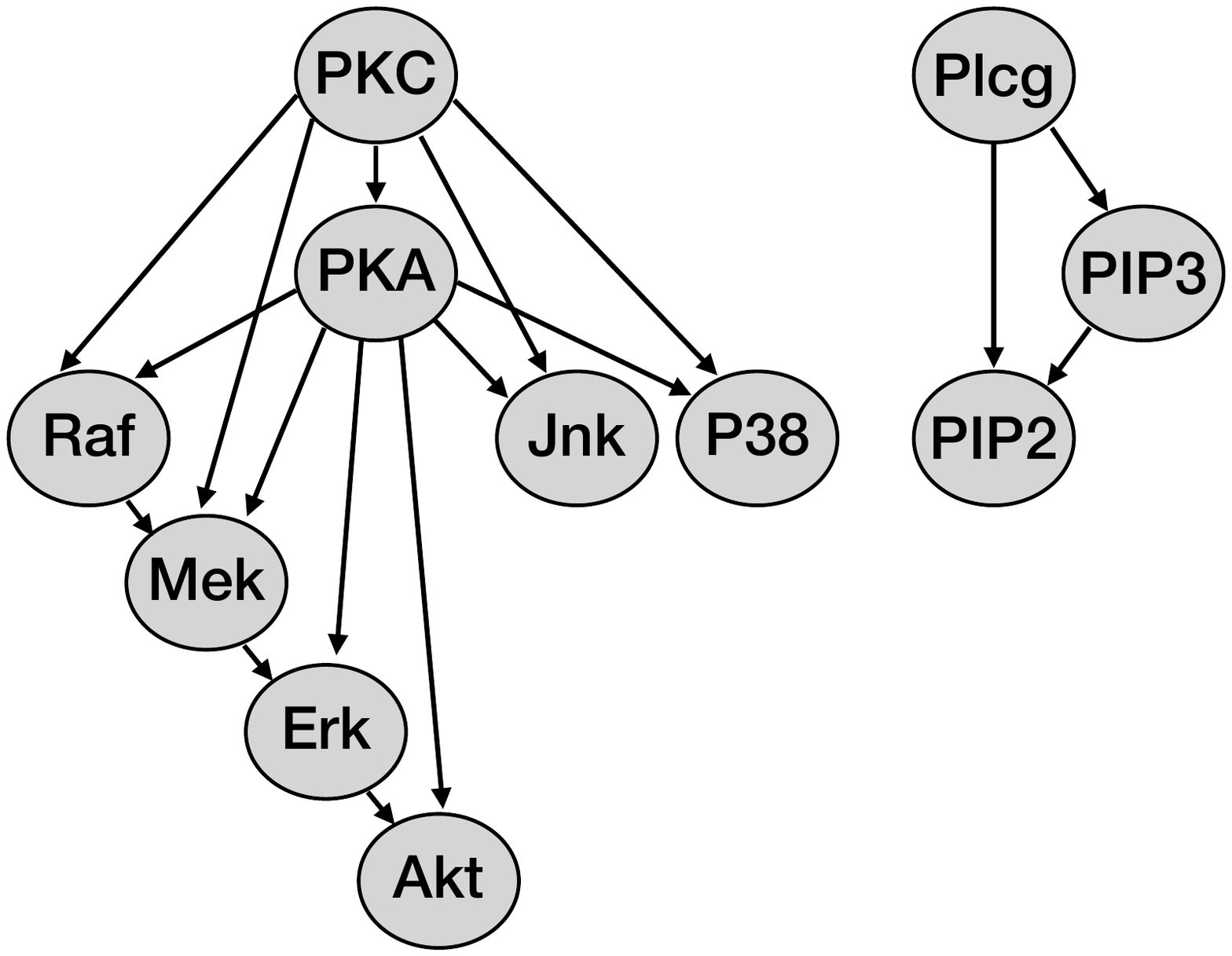}
        \caption{Structure of the Sachs BN \cite{bn_sachs}.}
        \label{fig:bn_sachs}
    \end{subfigure}
    \begin{subfigure}{.73\linewidth}
    \centering
    \begin{tabular}{llllll}
    \toprule
    Explanation & Output Nodes & d & Attack & AUC & Std. \\
    \midrule
    right sub. & Plcg,PIP3,PIP2 & 9 & \textbf{Bayes} & \textbf{0.773} & 0.097 \\
    right sub. & Plcg,PIP3,PIP2 & 9 & LRT & 0.746 & 0.101 \\
    right sub. & Plcg,PIP3,PIP2 & 9 & IP & 0.733 & 0.102 \\
    \hline
    leaf & Akt,Jnk,P38,PIP2 & 12 & \textbf{Bayes} & \textbf{0.856} & 0.085 \\
    leaf & Akt,Jnk,P38,PIP2 & 12 & LRT & 0.827 & 0.088 \\
    leaf & Akt,Jnk,P38,PIP2 & 12 & IP & 0.802 & 0.103 \\
    \hline
    leaf + root & PKC,Akt,Jnk,P38,Plcg,PIP2 & 18 & \textbf{Bayes} & \textbf{0.887} & 0.066 \\
    leaf + root & PKC,Akt,Jnk,P38,Plcg,PIP2 & 18 & LRT & 0.867 & 0.075 \\
    leaf + root & PKC,Akt,Jnk,P38,Plcg,PIP2 & 18 & IP & 0.831 & 0.080 \\
    \hline
    leaf + parent & PKA,Akt,Jnk,P38,Plcg,PIP3 & 18 & \textbf{Bayes} & \textbf{0.892} & 0.073 \\
    leaf + parent & PKA,Akt,Jnk,P38,Plcg,PIP3 & 18 & LRT & 0.862 & 0.076 \\
    leaf + parent & PKA,Akt,Jnk,P38,Plcg,PIP3 & 18 & IP & 0.819 & 0.084 \\
    \hline
    long path (L) & PKC,Raf,Mek,Erk,Akt & 15 & \textbf{Bayes} & \textbf{0.906} & 0.066 \\
    long path (L) & PKC,Raf,Mek,Erk,Akt & 15 & LRT & 0.851 & 0.095 \\
    long path (L) & PKC,Raf,Mek,Erk,Akt & 15 & IP & 0.824 & 0.107 \\
    \hline
    long path (R) & PKC,PKA,Mek,Erk,Akt & 15 & \textbf{Bayes} & \textbf{0.894} & 0.063 \\
    long path (R) & PKC,PKA,Mek,Erk,Akt & 15 & LRT & 0.853 & 0.084 \\
    long path (R) & PKC,PKA,Mek,Erk,Akt & 15 & IP & 0.834 & 0.085 
    \\
    \bottomrule
    \end{tabular}
    \caption{Mean attack AUC for different output nodes for the Sachs BN.}
    \label{fig:out_node_results}
    \end{subfigure}
    \caption{Impact of different output nodes for the Sachs BN. Output sets with more shared evidence leads to improved Bayes performance and more pronounced differences between attacks.}
    \label{fig:out_node_study}
\end{figure*}

\subsection{Impact of Output Nodes on Attack Performance}
The Bayesian approach provides better average AUC across the board, but we want to see how the number of output parameters and strength of their relationship impacts the performance gains provided by incorporating structural information into the attack. In Fig. \ref{fig:out_node_study}, we study the Sachs BN more closely, selecting different interesting sets of output nodes under the strong threat model. For output sets that contain only the leaf (or leaf and root) nodes of the BN, there is little direct dependency between output nodes, and the Bayesian approach boasts modest improvements, about 2-3\%. However, for output sets with strongly connected nodes, such as the leftmost longest path or the leaves with their shared parents, we see a significant improvement when the attack considers a Bayesian model of the population. Here, we see that the Bayesian method beats the LRT attack by 4-5\%, and beats the IP attack by 8-9\%.

This indicates that increasing the number of strongly connected nodes in a BN (and therefore increasing the dependency between attributes) means that the Bayesian approach is able to leverage more important information to improve its attack. Further, the results indicate that incorporating the attribute dependency structure in the attack is more effective when there are strong dependencies, regardless of the size of the output set.

\begin{figure*}
    \centering
    \begin{subfigure}{.4\linewidth}
        \centering
        \includegraphics[width=\linewidth]{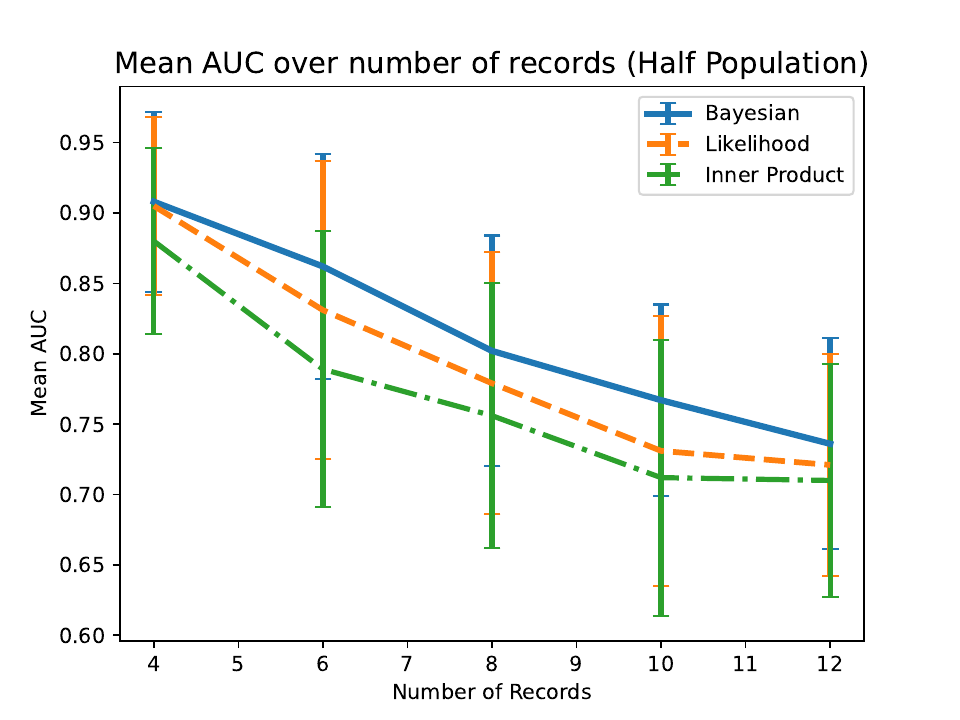}
        \caption{Half Repeated Population}
        \label{fig:vary_n_half}
    \end{subfigure}
    \begin{subfigure}{.4\linewidth}
        \centering
        \includegraphics[width=\linewidth]{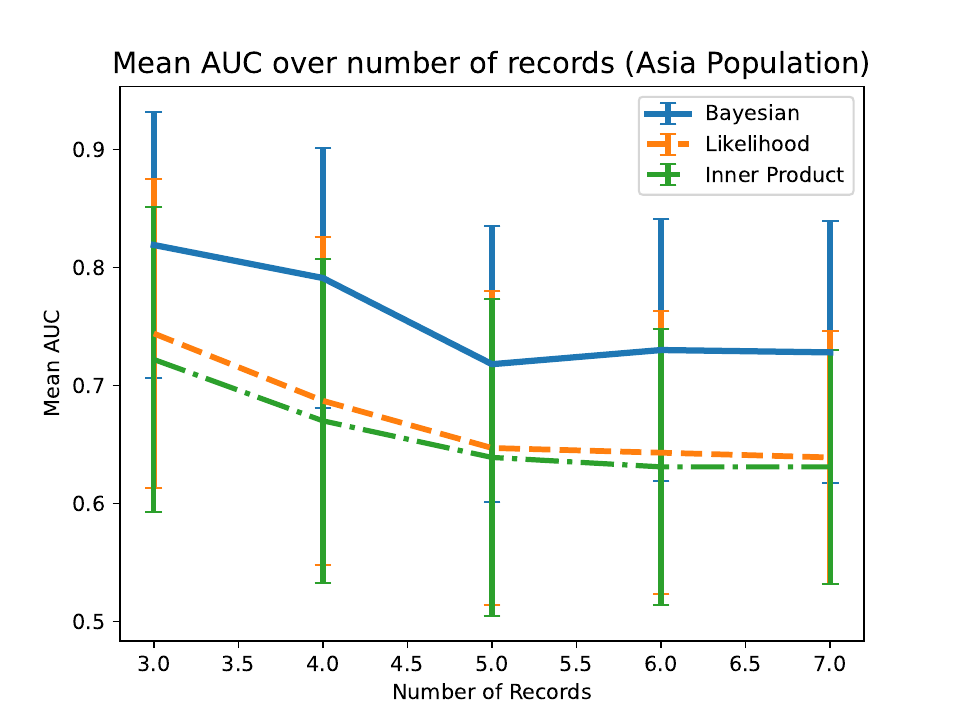}
        \caption{Asia Population (all nodes output)}
        \label{fig:vary_n_asia}
    \end{subfigure}
    \begin{subfigure}{.4\linewidth}
        \centering
        \includegraphics[width=\linewidth]{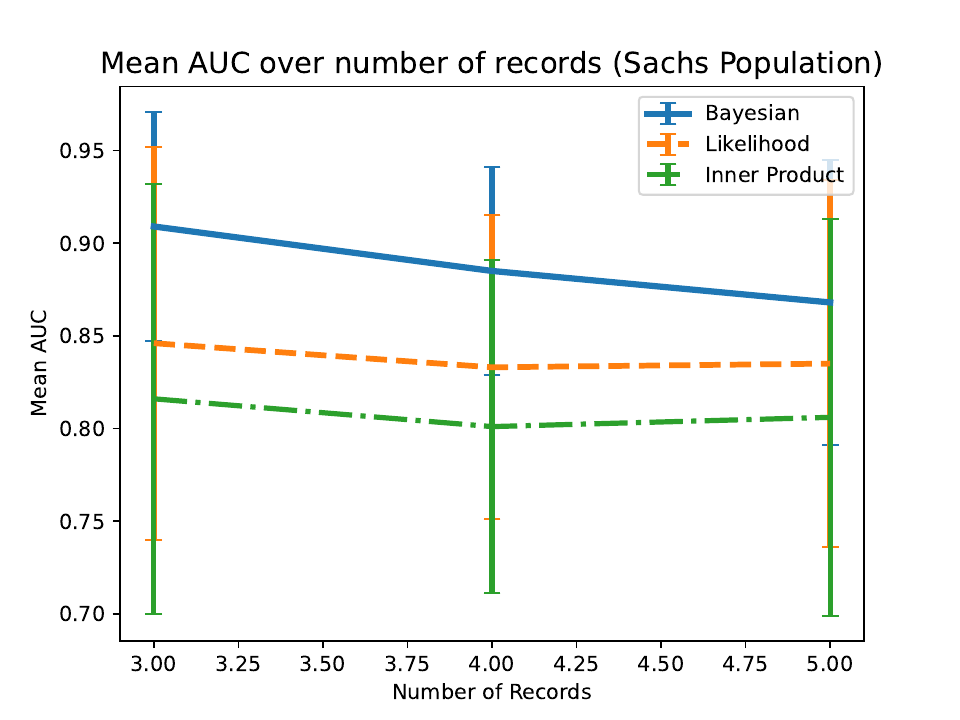}
        \caption{Sachs Population (longest path left output)}
        \label{fig:vary_n_sachs}
    \end{subfigure}
    \caption{Comparing attack performance as $\numrecords$ increases across different correlated populations. Values of $\numrecords$ chosen in relationship to the dimension of the population.}
    \label{fig:vary_n}
\end{figure*}
\subsection{Impact of Number of Samples on Attack Performance}\label{apdx:eval_num_samples}
To delve a bit further into the analysis of the performance trends of our approach, we wanted to understand how the number of records in the private dataset impacts the performance of the different attacks.

In Figure \ref{fig:vary_n}, we plot the mean AUC for different attacks over the number of records for the Half and Sachs populations. As expected, all attacks perform worse as the number of records grows. However, the differences between the three methods remains fairly constant. This means that incorporating information about the population results in an attack which deteriorates at a similar rate to the attacks which assume a product distributed population. Therefore, even as we increase the dimension of the data and the number of records, we still see gains from this Bayesian approach to the problem over the previous approaches.

Our technique is an exact method to compute the posterior distribution in full, which does not face pitfalls (with respect to outliers) associated with sampling. We see in Fig. \ref{fig:scale_n} that performance gains are consistent as $n$ increases for various sizes of Bayesian networks. This follows the theoretical result that the ratio between $d$ and $n$ is important, but the absolute size of $n$ is not.

\begin{figure*}
    \centering
    \begin{subfigure}{.45\linewidth}
        \centering
        \scalebox{0.9}{\begin{tabular}{llllll}
\toprule
\makecell{BN} & \makecell{\#\\Node} & \makecell{\#\\Param} & \makecell{\#\\ Out. \\Node} & \makecell{\#\\ Out. \\Param} & \makecell{Time\\(s)} \\
\midrule
product & 10 & 10 & 10 & 10 & 0.567 \\
product & 20 & 20 & 20 & 20 & 0.617 \\
product & 30 & 30 & 30 & 30 & 0.656 \\
\hline
half rep. & 10 & 10 & 10 & 10 & 0.551 \\
half rep. & 20 & 20 & 20 & 20 & 0.592 \\
half rep. & 30 & 30 & 30 & 30 & 0.614 \\
\hline
l/r rep. & 10 & 10 & 10 & 10 & 0.590 \\
l/r rep. & 20 & 20 & 20 & 20 & 2.473 \\
\hline
cancer & 5 & 10 & 5 & 10 & 0.669 \\
\hline
earthquake & 5 & 10 & 5 & 10 & 0.686 \\
\hline
asia & 8 & 16 & 8 & 16 & 0.710 \\
\hline
survey & 6 & 14 & 6 & 14 & 0.740 \\
\hline
sachs (leaves) & 11 & 33 & 4 & 12 & 0.884 \\
sachs (path) & 11 & 33 & 5 & 15 & 1.143 \\
sachs (leaf/rt) & 11 & 33 & 6 & 18 & 0.920 \\
\bottomrule
\end{tabular}}
        \caption{Mean computational for posterior using Roulette for different dimensions ($\numrecords=4$). For the Sachs Bayesian network, we use different sets of output nodes, so the number of output parameters is different from the number of parameters in the whole network.}
        \label{fig:scale_dimension}
    \end{subfigure}
    \qquad
    \hspace{2em}
    \begin{subfigure}{.34\linewidth}
        \centering
\scalebox{0.9}
{\includegraphics[width=\linewidth]{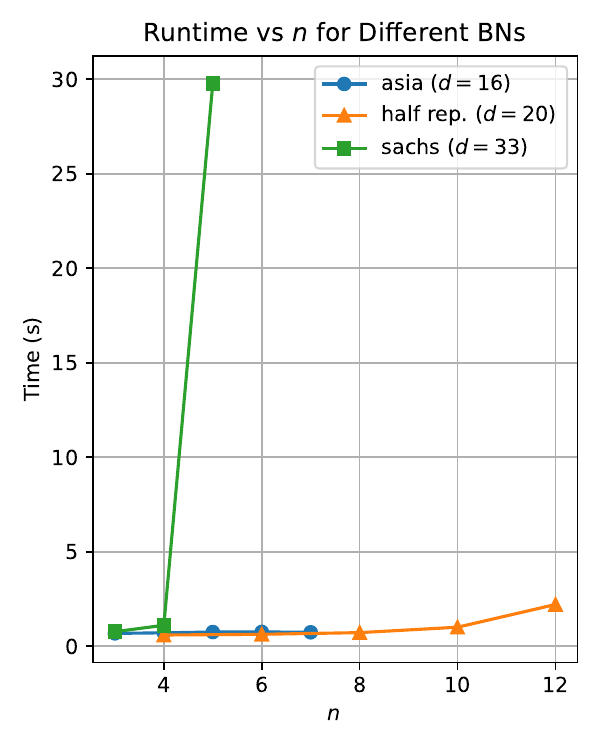}}

        \caption{Mean computational time for posterior using Roulette for different numbers of records. For the Sachs Bayesian network, we use the longest path (left) set of output nodes.}
        \label{fig:scale_n}
    \end{subfigure}
    \caption{Comparing exact posterior computation time in Roulette for various dimensions and numbers of records on various BNs.}
    \label{fig:scaling}
\end{figure*}

\subsection{Scalability}\label{sec:scalability}
To evaluate the scalability of our framework implementation, we look at how long it takes to compute the posterior in Algorithm \ref{alg:posterior} for individual target points as we increase the dimension and dependence between dimensions, as well as the number of records. Times listed are averaged over 20 unique private datasets, and 40 targets points (20 in, 20 out) for each dataset.

In Fig. \ref{fig:scale_dimension} we look at the scalability of our technique on different dimensions. We see that, for constant $\numrecords$, increasing the dimensionality of the Bayesian network results in sub-linear scaling of the runtime. In Fig. \ref{fig:scale_n} we look at the scalability of our technique on different numbers of records. We see that larger, more connected networks such as the left/right (l/r) repeated distribution and Sachs Bayesian network see significantly increasing runtimes as the number of records increases. 

While these results are promising, we find that Roulette has the hardest time scaling as we increase the number of samples in the private dataset, especially as the number of nodes in the Bayesian network increases. However, because of the trends we see in the AUC for the different approaches as we increase the number of samples and connected output nodes, we hypothesize that using larger dimensions will only amplify the gains that we get from using prior information about the population. These results are extremely promising, and indicate that there is a lot of utility to be gained by incorporating information about the dependency structure of the population into the attack framework.

Though this work does not yet scale to large examples, we have shown through theoretical and experimental results that using a Bayesian framework for designing an attack on a private dataset could have significant impact on attack utility when we have prior information. To compute larger $n$, one could use a PPL that supports approximate inference, though this will introduce complexity as there are many candidate approximate methods which require further tuning. Another interesting future direction would be to consider simplifications of the Bayesian network descriptions to attempt to improve the scalability of this approach on larger examples with more sample sizes, while still maintaining utility from the additional information gained from the Bayesian model of the data.

\section{Related Works}
\label{sec:related}
\textbf{Membership Inference.} The membership inference problem for statistical release has been applied to many areas. The primary attacks with strong theoretical results are likelihood ratio test (LRT) attack and its variants \cite{SankararamanGenomic2009,homer_resolving_2008,survey_2017} and the inner product (IP) attack \cite{inner_product_2015}, which typically only use the population marginals for analysis (see Sect. \ref{sec:classical_attack} and \ref{sec:inner_product_attack}). In contrast, our Bayesian approach leverages a complete model of the population in order to create a more adaptive attack, which we demonstrate recovers or outperforms the LRT and IP attacks.

\textbf{Defenses Against Membership Inference Attacks.} The primary defense against membership inference attacks for statistical release is to add random noise to the released statistics to make them \textit{differentially private} \cite{survey_2017,dwork2014textbook,dwork2006dp}. This technique is effective at limiting the utility of membership inference attacks, however there is a tradeoff between privacy and utility as the amount of noise increases. An interesting future direction of this work would be to evaluate how the Bayesian attacker responds to noise in comparison to the population marginal attacks when there is strong dependency between attributes, and how much randomness has to be applied (and thereby loss of accuracy) to offset the gains made by taking dependency structure into account.

\textbf{Bayesian Techniques for Privacy.} Bayesian inference techniques have been used in other capacities for evaluating privacy of systems. In \cite{oakley2024synthesizing}, Oakley \etal{} use exact Bayesian inference to verify differential privacy. In \cite{pardo2021privug}, Pardo \etal{} use probabilistic inference to debug privacy risks in program code. Techniques such as stochastic conditioning have also been proposed for incorporating into privacy attacks \cite{tolpin2021probabilistic}. There are also many techniques for making programming languages and program logics for differential privacy, but these mostly focus on verification \cite{near2019duet,barthe2016programming,gaboardi2013linear,reed2010distance}. In some instances, Bayesian networks are discussed in relationship to the membership inference problem. In \cite{bayesnetMI2021}, Murakonda \etal{} posit a problem setup where there are no assumptions on the population data. Instead of a statistic over the population, they are given a fixed graphical model (Bayesian network) with parameters learned from the population data. This contrasts our work, where the population is described as a Bayesian network, and the target private dataset is drawn from this Bayesian network.

\textbf{Probabilistic Inference and Probabilistic Programming.} Our implementation relies on the field of probabilistic programming for exact probabilistic inference. In this paper, we use the state-of-the-art PPL, Roulette \cite{moy2025roulette}, for efficiency and expressivity. Related languages \cite{holtzen2020scaling} also perform exact probabilistic inference and could be used instead, however they may not perform as well on larger examples. Other ongoing work on approximate Bayesian inference could be applied as well, though with potential performance tradeoffs \cite{sunnaaker2013approximate}. Recent work by Stites \etal{} also makes advances on hybrid methods of probabilistic inference, using both exact and approximate components \cite{stites2025multi}. In this work, we consider populations with discrete features and use discrete PPLs to compute the posterior, however there also exist PPLs for probabilistic inference on continuous random variables including Stan \cite{stan} and Pyro \cite{bingham2019pyro}.

\textbf{Machine Learning Privacy.}
Another related line of work is privacy attacks for machine learning, which uses statistical techniques for estimating privacy leakage in machine learning algorithms to find lower bounds~\cite{AuditingDP,nasr2021adversary,andrew2023oneshot,nasr2023tight,pillutla2023unleashing,steinke2023auditing,ye2022enhanced,zarifzadeh2024low,tong2025much,tao2025range,nasr2018comprehensive,shokri2017membership,kumar2020mlprivacy}. These methods efficiently analyze privacy over large datasets, but are applied to machine learning privacy rather than statistical release, and therefore have limited theoretical guarantees.

\section{Conclusions and Future Work}
\label{sec:conclusion}
This paper is a first step into a rich and largely unexplored area of Bayesian analysis of the membership inference problem. We present promising theoretical and experimental results that incorporating a Bayesian prior can create an adaptive attack that is more effective on populations with attribute dependencies than existing attacks that rely on only the population marginals. While our proposed method for computing the Bayesian posterior via exact probabilistic inference is highly effective, it is not the only approach. For example, approximate probabilistic inference methods could improve scalability, potentially at the expense of optimality and performance.

\bibliographystyle{IEEEtran}
\bibliography{refs.bib}

@article{survey_2017,
    title = {Exposed {A} {Survey} of {Attacks} on {Private} {Data}},
    volume = {4},
    issn = {2326-8298, 2326-831X},
    url = {https://www.annualreviews.org/doi/10.1146/annurev-statistics-060116-054123},
    doi = {10.1146/annurev-statistics-060116-054123},
    language = {en},
    number = {1},
    urldate = {2024-05-10},
    journal = {Annual Review of Statistics and Its Application},
    author = {Dwork, Cynthia and Smith, Adam and Steinke, Thomas and Ullman, Jonathan},
    month = mar,
    year = {2017},
    pages = {61--84},
}

@INPROCEEDINGS{inner_product_2015,
  author={Dwork, Cynthia and Smith, Adam and Steinke, Thomas and Ullman, Jonathan and Vadhan, Salil},
  booktitle={2015 IEEE 56th Annual Symposium on Foundations of Computer Science}, 
  title={Robust Traceability from Trace Amounts}, 
  year={2015},
  volume={},
  number={},
  pages={650-669},
  doi={10.1109/FOCS.2015.46}}

@article{SankararamanGenomic2009,
    title = {Genomic privacy and limits of individual detection in a pool},
    volume = {41},
    copyright = {2009 Springer Nature America, Inc.},
    issn = {1546-1718},
    url = {https://www.nature.com/articles/ng.436},
    doi = {10.1038/ng.436},
    language = {en},
    number = {9},
    urldate = {2024-08-15},
    journal = {Nature Genetics},
    author = {Sankararaman, Sriram and Obozinski, Guillaume and Jordan, Michael I. and Halperin, Eran},
    year = {2009},
    note = {Publisher: Nature Publishing Group},
    pages = {965--967},
}

@article{homer_resolving_2008,
    title = {Resolving {Individuals} {Contributing} {Trace} {Amounts} of {DNA} to {Highly} {Complex} {Mixtures} {Using} {High}-{Density} {SNP} {Genotyping} {Microarrays}},
    volume = {4},
    issn = {1553-7404},
    url = {https://journals.plos.org/plosgenetics/article?id=10.1371/journal.pgen.1000167},
    doi = {10.1371/journal.pgen.1000167},
    language = {en},
    number = {8},
    urldate = {2024-04-23},
    journal = {PLOS Genetics},
    author = {Homer, Nils and Szelinger, Szabolcs and Redman, Margot and Duggan, David and Tembe, Waibhav and Muehling, Jill and Pearson, John V. and Stephan, Dietrich A. and Nelson, Stanley F. and Craig, David W.},
    year = {2008},
    note = {Publisher: Public Library of Science},
    pages = {e1000167},
}

@inproceedings{bayesnetMI2021,
  title={Quantifying the privacy risks of learning high-dimensional graphical models},
  author={Murakonda, Sasi Kumar and Shokri, Reza and Theodorakopoulos, George},
  booktitle={International Conference on Artificial Intelligence and Statistics},
  pages={2287--2295},
  year={2021},
  organization={PMLR}
}

@inproceedings{ye2022enhanced,
  title={Enhanced Membership Inference Attacks against Machine Learning Models},
  author={Ye, Jiayuan and Maddi, Aadyaa and Murakonda, Sasi Kumar and Bindschaedler, Vincent and Shokri, Reza},
  booktitle={Proceedings of the 2022 ACM SIGSAC Conference on Computer and Communications Security},
  pages={3093--3106},
  year={2022}
}

@inproceedings{zarifzadeh2024low,
  title={Low-Cost High-Power Membership Inference Attacks},
  author={Zarifzadeh, Sajjad and Liu, Philippe and Shokri, Reza},
  booktitle={Forty-first International Conference on Machine Learning (ICML)},
  year={2024}
}

@inproceedings{tong2025much,
  title={How much of my dataset did you use? Quantitative Data Usage Inference in Machine Learning},
  author={Tong, Yao and Ye, Jiayuan and Zarifzadeh, Sajjad and Shokri, Reza},
  booktitle={The Thirteenth International Conference on Learning Representations (ICLR)},
  year={2025}
}

@inproceedings{tao2025range,
  title={Range Membership Inference Attacks},
  author={Tao, Jiashu and Shokri, Reza},
  booktitle={IEEE Conference on Secure and Trustworthy Machine Learning (SaTML)},
  year={2025}
}

@inproceedings{nasr2018comprehensive,
  title={Comprehensive privacy analysis of deep learning},
  author={Nasr, Milad and Shokri, Reza and Houmansadr, Amir},
  booktitle={Proceedings of the 2019 IEEE Symposium on Security and Privacy (SP)},
  pages={1--15},
  year={2018}
}

@inproceedings{shokri2017membership,
  title={Membership inference attacks against machine learning models},
  author={Shokri, Reza and Stronati, Marco and Song, Congzheng and Shmatikov, Vitaly},
  booktitle={2017 IEEE symposium on security and privacy (SP)},
  pages={3--18},
  year={2017},
  organization={IEEE}
}

@inproceedings{kumar2020mlprivacy,
  author    = {Sasi Kumar and Shokri, Reza},
  title     = {ML Privacy Meter: Aiding regulatory compliance by quantifying the privacy risks of machine learning},
  year      = {2020},
  maintitle = {The 20th Privacy Enhancing Technologies Symposium},
  booktitle = {Workshop on Hot Topics in Privacy Enhancing Technologies (HotPETs)},
}

@book{bn_cancer,
  title={Bayesian artificial intelligence},
  author={Korb, Kevin B and Nicholson, Ann E},
  year={2010},
  publisher={CRC press}
}

@Article{bnlearn,
    title = {Learning Bayesian Networks with the {bnlearn} {R}
      Package},
    author = {Marco Scutari},
    journal = {Journal of Statistical Software},
    year = {2010},
    volume = {35},
    number = {3},
    pages = {1--22},
    doi = {10.18637/jss.v035.i03},
}

@article{bn_asia,
  title={Local computations with probabilities on graphical structures and their application to expert systems},
  author={Lauritzen, Steffen L and Spiegelhalter, David J},
  journal={Journal of the Royal Statistical Society: Series B (Methodological)},
  volume={50},
  number={2},
  pages={157--194},
  year={1988},
  publisher={Wiley Online Library}
}

@book{bn_earthquake,
  title={Bayesian artificial intelligence},
  author={Korb, Kevin B and Nicholson, Ann E},
  year={2010},
  publisher={CRC press}
}

@article{bn_sachs,
  title={Causal protein-signaling networks derived from multiparameter single-cell data},
  author={Sachs, Karen and Perez, Omar and Pe'er, Dana and Lauffenburger, Douglas A and Nolan, Garry P},
  journal={Science},
  volume={308},
  number={5721},
  pages={523--529},
  year={2005},
  publisher={American Association for the Advancement of Science}
}

@book{bn_survey,
  title={Bayesian networks: with examples in R},
  author={Scutari, Marco and Denis, Jean-Baptiste},
  year={2021},
  publisher={Chapman and Hall/CRC}
}

@book{murphy_bayesian_decision,
  title={Probabilistic machine learning: an introduction},
  author={Murphy, Kevin P},
  year={2022},
  publisher={MIT press}
}

@book{adnan_param_learning,
  title={Modeling and reasoning with Bayesian networks},
  author={Darwiche, Adnan},
  year={2009},
  publisher={Cambridge university press}
}

@article{kitson2023structurelearning,
  title={A survey of Bayesian Network structure learning},
  author={Kitson, Neville Kenneth and Constantinou, Anthony C and Guo, Zhigao and Liu, Yang and Chobtham, Kiattikun},
  journal={Artificial Intelligence Review},
  volume={56},
  number={8},
  pages={8721--8814},
  year={2023},
  publisher={Springer}
}

@article{sunnaaker2013approximate,
  title={Approximate bayesian computation},
  author={Sunn{\aa}ker, Mikael and Busetto, Alberto Giovanni and Numminen, Elina and Corander, Jukka and Foll, Matthieu and Dessimoz, Christophe},
  journal={PLoS computational biology},
  volume={9},
  number={1},
  pages={e1002803},
  year={2013},
  publisher={Public Library of Science San Francisco, USA}
}

@article{pgmpy,
  author  = {Ankur Ankan and Johannes Textor},
  title   = {pgmpy: A Python Toolkit for Bayesian Networks},
  journal = {Journal of Machine Learning Research},
  year    = {2024},
  volume  = {25},
  number  = {265},
  pages   = {1--8},
  url     = {http://jmlr.org/papers/v25/23-0487.html}
}

@article{moy2025roulette,
  title={Roulette: A Language for Expressive, Exact, and Efficient Discrete Probabilistic Programming (with Appendices)},
  author={MOY, CAMERON and CZENSZAK, JACK and LI, JOHN M and MARSHALL, BRIANNA and HOLTZEN, STEVEN},
  journal={Proceedings of the ACM on Programming Languages},
  volume={9},
  number={PLDI},
  year={2025},
  publisher={ACM New York, NY, USA}
}

@article{holtzen2020scaling,
  title={Scaling exact inference for discrete probabilistic programs},
  author={Holtzen, Steven and Van den Broeck, Guy and Millstein, Todd},
  journal={Proceedings of the ACM on Programming Languages},
  volume={4},
  number={OOPSLA},
  pages={1--31},
  year={2020},
  publisher={ACM New York, NY, USA}
}

@article{stites2025multi,
  title={Multi-Language Probabilistic Programming},
  author={Stites, Sam and Li, John M and Holtzen, Steven},
  journal={Proceedings of the ACM on Programming Languages},
  volume={9},
  number={OOPSLA1},
  pages={1239--1266},
  year={2025},
  publisher={ACM New York, NY, USA}
}

@inproceedings{oakley2024synthesizing,
  title={Synthesizing Tight Privacy and Accuracy Bounds via Weighted Model Counting},
  author={Oakley, Lisa and Holtzen, Steven and Oprea, Alina},
  booktitle={2024 IEEE 37th Computer Security Foundations Symposium (CSF)},
  pages={449--463},
  year={2024},
  organization={IEEE}
}

@inproceedings{pardo2021privug,
  title={Privug: using probabilistic programming for quantifying leakage in privacy risk analysis},
  author={Pardo, Ra{\'u}l and Rafnsson, Willard and Probst, Christian W and Wasowski, Andrzej},
  booktitle={European Symposium on Research in Computer Security},
  pages={417--438},
  year={2021},
  organization={Springer}
}

@inproceedings{tolpin2021probabilistic,
  title={Probabilistic programs with stochastic conditioning},
  author={Tolpin, David and Zhou, Yuan and Rainforth, Tom and Yang, Hongseok},
  booktitle={International Conference on Machine Learning},
  pages={10312--10323},
  year={2021},
  organization={PMLR}
}

@inproceedings{reed2010distance,
  title={Distance makes the types grow stronger: a calculus for differential privacy},
  author={Reed, Jason and Pierce, Benjamin C},
  booktitle={Proceedings of the 15th ACM SIGPLAN international conference on Functional programming},
  pages={157--168},
  year={2010}
}

@inproceedings{gaboardi2013linear,
  title={Linear dependent types for differential privacy},
  author={Gaboardi, Marco and Haeberlen, Andreas and Hsu, Justin and Narayan, Arjun and Pierce, Benjamin C},
  booktitle={Proceedings of the 40th annual ACM SIGPLAN-SIGACT symposium on Principles of programming languages},
  pages={357--370},
  year={2013}
}

@article{near2019duet,
  title={Duet: an expressive higher-order language and linear type system for statically enforcing differential privacy},
  author={Near, Joseph P and Darais, David and Abuah, Chike and Stevens, Tim and Gaddamadugu, Pranav and Wang, Lun and Somani, Neel and Zhang, Mu and Sharma, Nikhil and Shan, Alex and others},
  journal={Proceedings of the ACM on Programming Languages},
  volume={3},
  number={OOPSLA},
  pages={1--30},
  year={2019},
  publisher={ACM New York, NY, USA}
}

@article{barthe2016programming,
  title={Programming language techniques for differential privacy},
  author={Barthe, Gilles and Gaboardi, Marco and Hsu, Justin and Pierce, Benjamin},
  journal={ACM SIGLOG News},
  volume={3},
  number={1},
  pages={34--53},
  year={2016},
  publisher={ACM New York, NY, USA}
}

@inproceedings{AuditingDP,
 author = {Jagielski, Matthew and Ullman, Jonathan and Oprea, Alina},
 booktitle = {Proceedings of Advances in Neural Information Processing Systems},
 series = {NeurIPS},
 pages = {22205--22216},
 title = {Auditing Differentially Private Machine Learning: How Private is Private {SGD?}},
 url = {https://proceedings.neurips.cc/paper/2020/file/fc4ddc15f9f4b4b06ef7844d6bb53abf-Paper.pdf},
 volume = {33},
 year = {2020}
}

@inproceedings{nasr2021adversary,
  author       = {Milad Nasr and
                  Shuang Song and
                  Abhradeep Thakurta and
                  Nicolas Papernot and
                  Nicholas Carlini},
  title        = {Adversary Instantiation: Lower Bounds for Differentially Private Machine
                  Learning},
  booktitle    = {42nd {IEEE} Symposium on Security and Privacy, {SP} 2021, San Francisco,
                  CA, USA, 24-27 May 2021},
  pages        = {866--882},
  publisher    = {{IEEE}},
  year         = {2021},
  url          = {https://doi.org/10.1109/SP40001.2021.00069},
  doi          = {10.1109/SP40001.2021.00069},
  bibsource    = {dblp computer science bibliography, https://dblp.org}
}

@article{andrew2023oneshot,
      title={One-shot Empirical Privacy Estimation for Federated Learning}, 
      author={Galen Andrew and Peter Kairouz and Sewoong Oh and Alina Oprea and H. Brendan McMahan and Vinith Suriyakumar},
      year={2023},
      journal   = {CoRR},
      volume    = {abs/2302.03098},
      archivePrefix={arXiv},
      primaryClass={cs.LG}
}

@inproceedings{nasr2023tight,
author = {Nasr, Milad and Hayes, Jamie and Steinke, Thomas and Balle, Borja and Tram\`{e}r, Florian and Jagielski, Matthew and Carlini, Nicholas and Terzis, Andreas},
title = {Tight auditing of differentially private machine learning},
year = {2023},
isbn = {978-1-939133-37-3},
publisher = {USENIX Association},
address = {USA},
booktitle = {Proceedings of the 32nd USENIX Conference on Security Symposium},
articleno = {92},
numpages = {18},
location = {Anaheim, CA, USA},
series = {SEC '23}
}

@inproceedings{pillutla2023unleashing,
title={Unleashing the Power of Randomization in Auditing Differentially Private {ML}},
author={Krishna Pillutla and Galen Andrew and Peter Kairouz and Hugh Brendan McMahan and Alina Oprea and Sewoong Oh},
booktitle={Thirty-seventh Conference on Neural Information Processing Systems},
year={2023},
url={https://openreview.net/forum?id=mlbes5TAAg}
}

@inproceedings{steinke2023auditing,
title={Privacy Auditing with One (1) Training Run},
author={Thomas Steinke, Milad Nasr, Matthew Jagielski},
booktitle={Thirty-seventh Conference on Neural Information Processing Systems},
year={2023},
url={https://openreview.net/forum?id=mlbes5TAAg}
}

@article{Columbia2025,
   author = {Weissman, Sara},
   title = {Hack at Columbia Hits 870K People},
   year = {2025},
   journal = {Inside Higher Education},
   note = {Available at: \url{https://www.insidehighered.com/news/tech-innovation/administrative-tech/2025/08/12/hack-columbia-university-hits-870k-people} (Accessed: {January 13th, 2025})},
   month = {12 Aug},
   type = {Newspaper Article}
}

@article{dwork2014textbook,
  title={The algorithmic foundations of differential privacy},
  author={Dwork, Cynthia and Roth, Aaron},
  journal={Foundations and trends{\textregistered} in theoretical computer science},
  volume={9},
  number={3-4},
  pages={211--487},
  year={2014},
  publisher={Emerald Publishing Limited}
}

@inproceedings{dwork2006dp,
  title={Calibrating noise to sensitivity in private data analysis},
  author={Dwork, Cynthia and McSherry, Frank and Nissim, Kobbi and Smith, Adam},
  booktitle={Theory of cryptography conference},
  pages={265--284},
  year={2006},
  organization={Springer}
}

@Misc{stan,
author =   {Stan Development Team},
title =    {{Stan Reference Manual}},
howpublished = {\url{https://mc-stan.org}},
year = {2026}
}

@article{bingham2019pyro,
  author    = {Eli Bingham and
               Jonathan P. Chen and
               Martin Jankowiak and
               Fritz Obermeyer and
               Neeraj Pradhan and
               Theofanis Karaletsos and
               Rohit Singh and
               Paul A. Szerlip and
               Paul Horsfall and
               Noah D. Goodman},
  title     = {Pyro: Deep Universal Probabilistic Programming},
  journal   = {J. Mach. Learn. Res.},
  volume    = {20},
  pages     = {28:1--28:6},
  year      = {2019},
  url       = {http://jmlr.org/papers/v20/18-403.html}
}

@inproceedings{machado2025sequential,
  title={Sequential conditional transport on probabilistic graphs for interpretable counterfactual fairness},
  author={Machado, Agathe Fernandes and Charpentier, Arthur and Gallic, Ewen},
  booktitle={Proceedings of the AAAI Conference on Artificial Intelligence},
  volume={39},
  number={18},
  pages={19358--19366},
  year={2025}
}

\appendix
\section{Generative AI Usage Statement}
The authors did not use LLMs in any way in the research or writing of this paper.
\section{Lemmas for Main Theorems}
\label{apdx:lemmas}
Here we present two main lemmas to help prove correctness of the Bayesian approach on selected data distributions.

\begin{lemma}[Ratio Equality]\label{lem:ratioeq}
    Given a population that follows a product distribution, the ratio $\bayesratio$ in eqn. (\ref{eqn:bayestest}) is equivalent to the ratio $\freqratio$ in eqn. (\ref{eqn:freq_test}).
\end{lemma}

\begin{proof}
    Let $\popbn$ be the Bayesian network of a product population with dimension $\dimension$, as described in Fig. \ref{fig:bn_product} with probability vector $\popexp$ such that $\popexp_j$ for $j\in\{1,\dots,\dimension\}$. Let private dataset $\dataset=\{\mathbf{x}^{(1)},\dots,\mathbf{x}^{(\numrecords)}\}$ consist of $\numrecords$ samples (records) drawn i.i.d. from $\population$ with sample mean $\datamarg=\frac{1}{\numrecords}\sum_{i=1}^{\numrecords} \mathbf{x}^{(i)}$. Let target $\target$ be some target record drawn from $\unifofset{\dataset}$ or from $\population$. For clarity, we set $A\gets (\popbn,\,\datamarg,\,\target)$.

    By the definition of conditional probability,
    \begin{align*}
        \bayesratio =\frac
    {\pr{\tlabelparam \mid A}}
    {\pr{\lnot\tlabelparam \mid A}}=\frac
    {\frac{\pr{\tlabelparam, A}}{\pr{A}}}
    {\frac{\pr{\lnot\tlabelparam,A}}{\pr{A}}}=\frac
    {\pr{\tlabelparam, A}}
    {\pr{\lnot\tlabelparam, A}}.
    \end{align*}

    By the definition of Algorithm \ref{alg:posterior}, the definition of the Bayesian network, and arithmetic, we have that 
    \begin{align*}
        \frac
    {\pr{\tlabelparam, A}}
    {\pr{\lnot\tlabelparam, A}}
    &=\frac
    {\pr{\target + \sum_{i=1}^{n-1}\popexp=\datamarg\cdot\numrecords}}
    {\pr{\sum_{i=1}^{n}\popexp=\datamarg\cdot\numrecords}}\\
    &=\frac
    {\pr{\sum_{i=1}^{n-1}\popexp=\datamarg\cdot\numrecords-\target}}
    {\pr{\sum_{i=1}^{n}\popexp=\datamarg\cdot\numrecords}}
    \end{align*}

    Because the coordinates are independent, we have
    \begin{align*}
        &\frac
    {\pr{\sum_{i=1}^{n-1}\popexp=\datamarg\cdot\numrecords-\target}}
    {\pr{\sum_{i=1}^{n}\popexp=\datamarg\cdot\numrecords}}\\
    &=\frac
    {\prod_{j=1}^{\dimension}{\pr{\sum_{i=1}^{n-1}\popexp_j=\datamarg_j\cdot\numrecords-\target_j}}}
    {\prod_{j=1}^{\dimension}\pr{\sum_{i=1}^{n}\popexp_j=\datamarg\cdot\numrecords}}\\
    &=\prod_{j=1}^{\dimension}
    \frac{\pr{\sum_{i=1}^{n-1}\popexp_j=\datamarg_j\cdot\numrecords-\target_j}}
    {\pr{\sum_{i=1}^{n}\popexp_j=\datamarg\cdot\numrecords}}
    \end{align*}
    
    Let's look at a single coordinate, $j\in\{1,\dots,\dimension\}$ and let $k=\datamarg_j\cdot \numrecords$ and $q=1-\popexp_j$ for notational simplicity. Because we have binomial distribution ($\bindist{\cdot}{\cdot}$) on the top and bottom, therefore,  we have that 
    \begin{align*}
        \frac
    {\pr{\sum_{i=1}^{n-1}\popexp_j=k-\target_j}}
    {\pr{\sum_{i=1}^{n}\popexp_j=k}}
    &=\frac
    {\bindist{\numrecords-1}{\popexp_j}(k-\target_j)}
    {\bindist{\numrecords}{\popexp_j}(k)}\\
    &=\frac
    {{\numrecords-1 \choose k-\target_j} \popexp_j^{k-\target_j}(q^{\numrecords-1-(k-\target_j)})}
    {{\numrecords \choose k} \popexp_j^{k}(q^{n-k})}
    \end{align*}

    If $\target_j=0$, we have
    \begin{align*}
        \frac
    {{\numrecords-1 \choose k-\target_j} \popexp_j^{k-\target_j}(q^{\numrecords-1-(k-\target_j)})}
    {{\numrecords \choose k} \popexp_j^{k}(q^{n-k})}
        &=\frac
    {{\numrecords-1 \choose k} \popexp_j^{k}(q^{\numrecords-1-(k)})}
    {{\numrecords \choose k} \popexp_j^{k}(q^{n-k})}\\
        &=\frac
    {{\numrecords-k}}
    {{\numrecords}q}
        =\frac{\numrecords-(\numrecords\datamarg_j)}{\numrecords(1-\popexp_j)}\\
        &=\frac{\numrecords(1-\datamarg_j)}{\numrecords(1-\popexp_j)}
        =\frac{1-\datamarg_j}{1-\popexp_j}
    \end{align*}

    If $\target_j=1$, we have
    \begin{align*}
        \frac
    {{\numrecords-1 \choose k-\target_j} \popexp_j^{k-\target_j}(q^{\numrecords-1-(k-\target_j)})}
    {{\numrecords \choose k} \popexp_j^{k}(q^{n-k})}
        &=\frac
    {{\numrecords-1 \choose k-1} \popexp_j^{k-1}(q^{\numrecords-1-(k-1)})}
    {{\numrecords \choose k} \popexp_j^{k}(q^{n-k})}\\
        &=\frac{k}{\numrecords\popexp_j}=\frac{\numrecords\datamarg_j}{\numrecords\popexp_j}=\frac{\datamarg_j}{\popexp_j}
    \end{align*}

    Therefore, we have 

    \begin{align*}
        \bayesratio
        &=\prod_{j=1}^{\dimension}
        \frac{\pr{\sum_{i=1}^{n-1}\popexp_j=\datamarg_j\cdot\numrecords-\target_j}}
        {\pr{\sum_{i=1}^{n}\popexp_j=\datamarg_j\cdot\numrecords}}\\
        &=\prod_{j=1}^{\dimension}\left\{\text{if } \target_j=1 \text{ then }\frac{\datamarg_j}{\popexp_j} \text{ else }\frac{1-\datamarg_j}{1-\popexp_j}\right\}\\
        &=\frac
        {\prod_{j=1}^{\dimension}\left\{\text{if } \target_j=1 \text{ then }\datamarg_j \text{ else }1-\datamarg_j\right\}}
        {\prod_{j=1}^{\dimension}\left\{\text{if } \target_j=1 \text{ then }\popexp_j \text{ else }1-\popexp_j\right\}}\\
        &=\Lambda,
    \end{align*}
    as desired.
\end{proof}

\begin{lemma}[Partial Ratio Equality]\label{lem:partialratioeq}
    Given a population that follows a left/right repeated distribution as described in Fig. \ref{fig:bn_side} with parameter $\midpt$, the ratio $\bayesratio$ in eqn. (\ref{eqn:bayestest}) is equivalent to the ratio $\freqratio_{\midpt,\side}$ in eqn. (\ref{eqn:freq_test_side_left}) or (\ref{eqn:freq_test_side_right}) for fixed $\side\in\{left,\,right\}$.
\end{lemma}

\begin{proof}
    Let $\popbn$ be the Bayesian network of a left/right repeated population with dimension $\dimension$ and midpoint $\midpt$, as described in Fig. \ref{fig:bn_side}. Let private dataset $\dataset=\{\mathbf{x}^{(1)},\dots,\mathbf{x}^{(\numrecords)}\}$ consist of $\numrecords$ samples (records) drawn i.i.d. from $\population$ with sample mean $\datamarg=\frac{1}{\numrecords}\sum_{i=1}^{\numrecords} \mathbf{x}^{(i)}$. Let target $\target$ be some target record drawn from $\unifofset{\dataset}$ or from $\population$. For clarity, we set $A\gets (\popbn,\,\datamarg,\,\target)$.

    By the definition of conditional probability,
\begin{align*}
        \bayesratio =\frac
    {\pr{\tlabelparam \mid A}}
    {\pr{\lnot\tlabelparam \mid A}}=\frac
    {\frac{\pr{\tlabelparam, A}}{\pr{A}}}
    {\frac{\pr{\lnot\tlabelparam,A}}{\pr{A}}}=\frac
    {\pr{\tlabelparam, A}}
    {\pr{\lnot\tlabelparam, A}}.
    \end{align*}

    If $\side=right$, we know that attributes $1$ to $\midpt$ are independent, and for $j\in\{\midpt+1,\dots,\dimension\}$, $\popexp_j=\popexp_{\midpt}$, $\datamarg_j=\datamarg_{\midpt}$, and $\target_j=\target_{\midpt}$. Therefore we can split this quantity into two parts:
    \begin{align*}
        \frac
    {\pr{\tlabelparam, A}}
    {\pr{\lnot\tlabelparam, A}}=&\\
    &\prod_{j=1}^{\midpt-1}
    \left(\frac{\pr{\sum_{i=1}^{n-1}\popexp_j=\datamarg_j\cdot\numrecords-\target_j}}{\pr{\sum_{i=1}^{n}\popexp_j=\datamarg_j\cdot \numrecords}}\right)\\
    &\cdot
    \frac
    {\pr{\sum_{i=1}^{n-1}\mathbf{p}_{\midpt}=\datamarg_{\midpt}\cdot\numrecords-\target_{\midpt}}}{\pr{\sum_{i=1}^{n}\mathbf{p}_{\midpt}=\datamarg_{\midpt}\cdot\numrecords}}.
    \end{align*}
    
    By applying Lemma \ref{lem:ratioeq} twice, once with the product network of dimension $\dimension=\midpt$ and once with the product network of dimension $\dimension=1$, this is equivalent to 
    \begin{align*}
     &\prod_{j=1}^{\midpt-1}
     \frac
        {\left\{\text{if } \target_j=1 \text{ then }\datamarg_j \text{ else }1-\datamarg_j\right\}}
        {\left\{\text{if } \target_j=1 \text{ then }\popexp_j \text{ else }1-\popexp_j\right\}}\\
        &\cdot\frac
        {\left\{\text{if } \target_{\midpt}=1 \text{ then }\datamarg_{\midpt} \text{ else }1-\datamarg_{\midpt}\right\}}
        {\left\{\text{if } \target_{\midpt}=1 \text{ then }\popexp_{\midpt} \text{ else }1-\popexp_{\midpt}\right\}}\\
    &=\frac
        {\prod_{j=1}^{\midpt}\left\{\text{if } \target_j=1 \text{ then }\datamarg_j \text{ else }1-\datamarg_j\right\}}
        {\prod_{j=1}^{\midpt}\left\{\text{if } \target_j=1 \text{ then }\popexp_j \text{ else }1-\popexp_j\right\}}\\
    &=\freqratio_{\midpt,\side},
    \end{align*}
    as desired.

    When $\side=left$, we have a symmetrical argument, where the left hand of the product represents the repeated elements for $j=1$ to $\midpt-1$, and the right hand side of the product represents the independent elements $\midpt$ to $d$. We similarly apply Lemma \ref{lem:ratioeq} to both sides of the product and get:
    \begin{align*}
        \bayesratio=&\frac
        {\left\{\text{if } \target_{\midpt}=1 \text{ then }\datamarg_{\midpt} \text{ else }1-\datamarg_{\midpt}\right\}}
        {\left\{\text{if } \target_{\midpt}=1 \text{ then }\popexp_{\midpt} \text{ else }1-\popexp_{\midpt}\right\}}\\
     &\cdot\prod_{j=\midpt-1}^{\dimension}
     \frac
        {\left\{\text{if } \target_j=1 \text{ then }\datamarg_j \text{ else }1-\datamarg_j\right\}}
        {\left\{\text{if } \target_j=1 \text{ then }\popexp_j \text{ else }1-\popexp_j\right\}}\\
    &=\frac
        {\prod_{j=\midpt}^{\dimension}\left\{\text{if } \target_j=1 \text{ then }\datamarg_j \text{ else }1-\datamarg_j\right\}}
        {\prod_{j=\midpt}^{\dimension}\left\{\text{if } \target_j=1 \text{ then }\popexp_j \text{ else }1-\popexp_j\right\}}\\
    &=\freqratio_{\midpt,\side},
    \end{align*}
    as desired.

    We have covered both cases, thus completing our proof.
    
\end{proof}
\section{Example Bayesian Network Construct}\label{apdx:network_ppl}
This is an example of an encoding of the Cancer Bayesian network from Fig. \ref{fig:bn_cancer} written in the Roulette probabilistic programming language.
\begin{Verbatim}
(define NETWORK
  '((variable Pollution (type discrete (2) (low high)))
    (variable Smoker (type discrete (2) (True False)))
    (variable Cancer (type discrete (2) (True False)))
    (variable Xray (type discrete (2) (positive negative)))
    (variable Dyspnoea (type discrete (2) (True False)))
    (probability (Cancer Pollution Smoker) 
    ((high False) 0.02 0.98)
    ((high True) 0.05 0.95)
    ((low False) 0.001 0.999)
    ((low True) 0.03 0.97))
(probability (Dyspnoea Cancer) 
    ((False) 0.3 0.7)
    ((True) 0.65 0.35))
(probability (Pollution) (table 0.9 0.1))
(probability (Smoker) (table 0.3 0.7))
(probability (Xray Cancer) 
    ((False) 0.2 0.8)
    ((True) 0.9 0.1))))
\end{Verbatim}
\section{Additional Experimental Results for Weak and Weakest Threat Models}\label{apdx:weak_weakest}
In Sec. \ref{sec:weaker} we present the results for running the attacks under the weak and weakest threat models. 
In Fig. \ref{fig:all_weak_weakest}, we demonstrate that the Bayesian attacker is similarly effective across three other benchmarking Bayesian networks.

\begin{figure*}
    \centering
    \begin{subfigure}{.4\linewidth}
        \centering
        \scalebox{0.9}{\begin{tabular}{llllll}
\toprule
BN (\dimension) & $m$ & Attack  & AUC & Std. \\
\midrule
cancer (10) & 10 & \textbf{Bayes} & 0.751000 & 0.104000 \\
cancer (10) & 10 & LRT & 0.710000 & 0.123000 \\
cancer (10) & 10 & IP & 0.695000 & 0.127000 \\
\midrule
cancer (10) & 50 & \textbf{Bayes} & 0.750000 & 0.107000 \\
cancer (10) & 50 & LRT & 0.708000 & 0.122000 \\
cancer (10) & 50 & IP & 0.693000 & 0.127000 \\
\midrule
cancer (10) & 100 & \textbf{Bayes} & 0.748000 & 0.109000 \\
cancer (10) & 100 & LRT & 0.707000 & 0.124000 \\
cancer (10) & 100 & IP & 0.693000 & 0.128000 \\
\midrule
earthquake (10) & 10 & \textbf{Bayes} & 0.579000 & 0.065000 \\
earthquake (10) & 10 & LRT & 0.539000 & 0.081000 \\
earthquake (10) & 10 & IP & 0.522000 & 0.082000 \\
\midrule
earthquake (10) & 50 & \textbf{Bayes} & 0.579000 & 0.067000 \\
earthquake (10) & 50 & LRT & 0.532000 & 0.086000 \\
earthquake (10) & 50 & IP & 0.517000 & 0.084000 \\
\midrule
earthquake (10) & 100 & \textbf{Bayes} & 0.573000 & 0.064000 \\
earthquake (10) & 100 & LRT & 0.536000 & 0.077000 \\
earthquake (10) & 100 & IP & 0.520000 & 0.076000 \\
\midrule
survey (14) & 10 & \textbf{Bayes} & 0.856000 & 0.086000 \\
survey (14) & 10 & LRT & 0.854000 & 0.086000 \\
survey (14) & 10 & IP & 0.829000 & 0.093000 \\
\midrule
survey (14) & 50 & \textbf{Bayes} & 0.855000 & 0.089000 \\
survey (14) & 50 & LRT & 0.852000 & 0.089000 \\
survey (14) & 50 & IP & 0.824000 & 0.096000 \\
\midrule
survey (14) & 100 & \textbf{Bayes} & 0.855000 & 0.087000 \\
survey (14) & 100 & LRT & 0.853000 & 0.088000 \\
survey (14) & 100 & IP & 0.826000 & 0.096000 \\
\bottomrule
\end{tabular}}
        \caption{Public Auxiliary Dataset and Structure (Weak Attacker)}
        \label{fig:all_weak}
    \end{subfigure}
    \qquad
    \hspace{2em}
    \begin{subfigure}{.4\linewidth}
        \centering
\scalebox{0.9}{\begin{tabular}{lllll}
\toprule
BN (\dimension) & $m$ & Attack  & AUC & Std. \\
\midrule
cancer (10) & 10 &  \textbf{Bayes} & 0.751000 & 0.110000 \\
cancer (10) & 10 &  LRT & 0.705000 & 0.129000 \\
cancer (10) & 10 &  IP & 0.688000 & 0.135000 \\
\midrule
cancer (10) & 50 & \textbf{Bayes} & 0.760000 & 0.106000 \\
cancer (10) & 50 & LRT & 0.719000 & 0.126000 \\
cancer (10) & 50 & IP & 0.704000 & 0.130000 \\
\midrule
cancer (10) & 100 & \textbf{Bayes} & 0.748000 & 0.108000 \\
cancer (10) & 100 & LRT & 0.700000 & 0.128000 \\
cancer (10) & 100 & IP & 0.684000 & 0.134000 \\
\midrule
earthquake (10) & 10 &  \textbf{Bayes} & 0.573000 & 0.065000 \\
earthquake (10) & 10 &  LRT & 0.538000 & 0.078000 \\
earthquake (10) & 10 &  IP & 0.527000 & 0.078000 \\
\midrule
earthquake (10) & 50 & \textbf{Bayes} & 0.577000 & 0.069000 \\
earthquake (10) & 50 & LRT & 0.535000 & 0.084000 \\
earthquake (10) & 50 & IP & 0.521000 & 0.084000 \\
\midrule
earthquake (10) & 100 & \textbf{Bayes} & 0.574000 & 0.066000 \\
earthquake (10) & 100 & LRT & 0.534000 & 0.081000 \\
earthquake (10) & 100 & IP & 0.523000 & 0.080000 \\
\midrule
survey (14) & 10 &  \textbf{Bayes} & 0.856000 & 0.088000 \\
survey (14) & 10 &  LRT & 0.850000 & 0.090000 \\
survey (14) & 10 &  IP & 0.823000 & 0.099000 \\
\midrule
survey (14) & 50 & \textbf{Bayes} & 0.860000 & 0.088000 \\
survey (14) & 50 & LRT & 0.858000 & 0.088000 \\
survey (14) & 50 & IP & 0.834000 & 0.095000 \\
\midrule
survey (14) & 100 & \textbf{Bayes} & 0.853000 & 0.091000 \\
survey (14) & 100 & LRT & 0.848000 & 0.093000 \\
survey (14) & 100 & IP & 0.824000 & 0.097000 \\
\bottomrule
\end{tabular}}
        \caption{Auxiliary Dataset and No Structure (Weakest Attacker)}
        \label{fig:all_weakest}
    \end{subfigure}
    \caption{Comparing our Bayesian (Bayes) attack AUC with the likelihood ratio test (LRT) and inner product (IP) attacks under the weak and weakest threat models on the three benchmarking Bayesian networks.}
    \label{fig:all_weak_weakest}
\end{figure*}

\end{document}